\documentclass{lmcs}
\pdfoutput=1

\usepackage{lastpage}
\lmcsdoi{18}{1}{39}
\lmcsheading{}{\pageref{LastPage}}{}{}%
{Apr.~13,~2021}{Mar.~11,~2022}{}


\usepackage{amsmath,amsthm,amssymb,hyperref,cleveref}
\usepackage{contents/macros}
\usepackage[utf8]{inputenc}

\graphicspath{{./contents/graphics/}}

\begin{document}

\title[A Scalable Algorithm for Decentralized Actor Termination Detection]{A Scalable Algorithm for Decentralized \texorpdfstring{\\}{} Actor Termination Detection}
\titlecomment{A previous version of this article appeared at CONCUR 2020.}

\author[D.~Plyukhin]{Dan Plyukhin}
\address{University of Illinois at Urbana-Champaign, USA}
\email{\{\texttt{daniilp2},\texttt{agha}\}\texttt{@illinois.edu}}

\author[G.~Agha]{Gul Agha}

\begin{abstract}
Automatic {\em garbage collection\/} (GC) prevents certain kinds of
bugs and reduces programming overhead.  GC techniques for sequential
programs are based on {\em reachability analysis\/}.
However, testing reachability from a root set is inadequate for determining whether an {\em actor\/} is garbage: Observe that an unreachable actor 
may send a message to a reachable actor.  
Instead, it is sufficient to check \emph{termination} (sometimes also called \emph{quiescence}): an
actor is terminated if it is not currently processing a message 
and cannot receive a message in the future.
Moreover, many actor frameworks provide all
actors with access to file I/O or external storage; without inspecting an actor's internal code, it is necessary to check that the actor has terminated to ensure that it may be garbage collected in these frameworks.
Previous algorithms to detect actor garbage require coordination
mechanisms such as causal message delivery or nonlocal monitoring of
actors for mutation.  Such coordination mechanisms adversely affect
concurrency and are therefore expensive in distributed systems.  We
present a low-overhead {\em deferred reference listing\/} technique (called
{\em DRL\/}) for termination detection in actor systems.  DRL is based
on asynchronous local snapshots and message-passing between actors.
This enables a decentralized implementation and transient network
partition tolerance.  The paper provides a formal description of DRL,
shows that all actors identified as garbage have indeed terminated
(safety), and that all terminated actors---under certain reasonable
assumptions---will eventually be identified (liveness).
\end{abstract}

\maketitle

\section{Introduction}

The actor model~\cite{books/daglib/0066897,journals/cacm/Agha90} is a
foundational model of concurrency that has been widely adopted for its
scalability: for example, actor languages have been used to implement services at
PayPal~\cite{PayPalBlowsBillion}, Discord~\cite{vishnevskiyHowDiscordScaled2017}, and in the United Kingdom's National Health Service
database~\cite{NHSDeployRiak2013}.  In the actor model, stateful processes known as \emph{actors}
execute concurrently and communicate by
sending asynchronous messages to other actors, provided they have a \emph{reference} (also called a \emph{mail address} or \emph{address} in the literature) to the recipient. Actors can also spawn new actors. An actor is said to be \emph{garbage} if it can be destroyed without affecting the system's observable behavior.

Although a number of algorithms for automatic actor garbage collection (GC) have been proposed \cite{ clebschFullyConcurrentGarbage2013,
  kafuraConcurrentDistributedGarbage1995,
  vardhanUsingPassiveObject2003,
  venkatasubramanianScalableDistributedGarbage1992,
  wangConservativeSnapshotbasedActor2011,
  wangDistributedGarbageCollection2006}, actor languages and
frameworks currently popular in industry (such as Akka \cite{Akka},
Erlang \cite{armstrongConcurrentProgrammingERLANG1996}, and Orleans
\cite{bykovOrleansCloudComputing2011}) require that programmers
garbage collect actors manually.  We believe this is because the GC algorithms
proposed thus far are too expensive to implement in distributed
systems.  In order to find applicability in real-world actor runtimes, we argue that a GC algorithm should satisfy the following properties:

\begin{enumerate}
\item (\emph{Low latency}) GC should not restrict concurrency in the
  application.
\item (\emph{High throughput}) GC should not impose significant space
  or message overhead.
\item (\emph{Scalability}) GC should scale with the number of actors and nodes in the system.
\end{enumerate} 
To the best of our knowledge, no previous algorithm satisfies all
three constraints. The first requirement precludes any global
synchronization between actors, a ``stop-the-world'' step, or a
requirement for causal order delivery of all messages.  The second
requirement means that the number of additional ``control'' messages imposed by the
algorithm should be minimal.  The third requirement precludes
algorithms based on global snapshots, since that requires all actors to respond before any garbage can be collected; such a delay may become unacceptable as a system grows large.

To address these goals, we have developed a garbage collection
technique called \emph{DRL} for \emph{Deferred Reference Listing}. The primary advantage of DRL is that it is decentralized and incremental: local garbage can be collected in a subsystem without communicating with the rest of the system. Systems can also cooperate to detect distributed garbage by exchanging minimal amounts of information. Garbage collection can be performed concurrently with the application and imposes no message ordering constraints. We also expect DRL to be reasonably efficient in practice, since it does not require many additional messages or significant actor-local computation.

DRL works as follows.  The \emph{communication protocol} (\Cref{sec:model}) tracks information, such as references and message counts, and stores it in each actor's state.  Actors periodically send out copies of their local state (called \emph{snapshots}) to be stored at one or more designated \emph{snapshot aggregator} actors. Each aggregator periodically searches its local store to find a subset of snapshots representing terminated actors (\Cref{sec:termination-detection,sec:maximal}).  Once an actor is determined to have terminated, it can be garbage collected by, for example, sending it a \emph{self-destruct} message. We prove that non-terminated actors will never be garbage collected (Corollary~\ref{cor:safety}). Moreover, if every terminated actor eventually sends a snapshot to the aggregator, then all terminated actors will eventually be detected (\Cref{thm:liveness}).

Since DRL is defined in terms of the actor model, it is oblivious to details of a particular implementation (such as how sequential computations are represented or where actors are located). Our technique is therefore applicable to different actor frameworks; in particular, it may be
implemented as a library.  Moreover, it can also be applied to open
systems, allowing an actor system using DRL to interoperate with a manually-collected actor
system.

The outline of the paper is as follows. We provide a characterization of actor garbage in \Cref{sec:background} and discuss related work in \Cref{sec:related-work}. We then provide a
specification of the DRL protocol in \Cref{sec:model}.
In \Cref{sec:chain-lemma}, we describe a key property of DRL called the \emph{Chain Lemma}. In \Cref{sec:termination-detection}, we use this lemma to define when a set of snapshots is \emph{finalized} and prove that finalized snapshots correspond to terminated actors. In \Cref{sec:maximal}, we give algorithms that allow snapshot aggregators to find finalized subsets in an arbitrary set of snapshots. In \Cref{sec:coop}, we describe an efficient protocol for snapshot aggregators to detect distributed cycles of terminated actors.
We conclude in \Cref{sec:future-work} with some discussion of future work and how DRL may be used in practice.

\section{Preliminaries}
\label{sec:background}

An actor can only receive a message when it is \emph{idle}. Upon
receiving a message, it becomes \emph{busy}. A busy actor can perform
an unbounded sequence of \emph{actions} before becoming idle. In~\cite{aghaFoundationActorComputation1997}, an action may be to spawn an actor, send a message, or perform a (local) computation. The actor model places no constraints on message delivery order, though a particular implementation could provide stronger guarantees. We will also assume that actors can perform effects, such as file I/O. The actions an actor performs in response to a message are dictated by its application-level code, called a \emph{behavior}.

Actors can also receive messages from \emph{external} actors (such as the user) by
becoming \emph{receptionists}. An actor $A$ becomes a receptionist when its address is exposed to an external actor. Subsequently, any external actor can potentially obtain $A$'s address and send it a message.

\begin{figure} \centering \begin{tikzpicture}
	\begin{pgfonlayer}{nodelayer}
		\node [style=none] (6) at (-4.75, -2) {$m$};
		\node [style=actor] (51) at (-8, -6) {$B$};
		\node [style=busyActor] (53) at (-10, -4) {$A$};
		\node [style=actor] (54) at (-6, -4) {$C$};
		\node [style=terminatedActor] (55) at (-2, -2) {$F$};
		\node [style=actor] (56) at (-4, -6) {$D$};
		\node [style=actor] (57) at (-2, -4) {$E$};
		\node [style=busyActor] (58) at (-10, -9) {$A$};
		\node [style=none] (59) at (-10, -10.325) {Busy};
		\node [style=none] (60) at (4.25, -9) {};
		\node [style=none] (61) at (5.75, -9) {};
		\node [style=none] (62) at (9.25, -9) {};
		\node [style=none] (63) at (10.75, -9) {};
		\node [style=none] (64) at (5, -10.25) {};
		\node [style=none] (65) at (5, -10.275) {Reference};
		\node [style=none] (66) at (10, -10.35) {Message};
		\node [style=none] (68) at (10, -8.75) {$m$};
		\node [style=actor] (69) at (-5, -9) {$A$};
		\node [style=none] (70) at (-5, -10.25) {Idle};
		\node [style=actor] (72) at (6, -6) {$B$};
		\node [style=busyActor] (73) at (4, -4) {$A$};
		\node [style=busyActor] (74) at (8, -4) {$C$};
		\node [style=terminatedActor] (75) at (12, -2) {$F$};
		\node [style=actor] (76) at (10, -6) {$D$};
		\node [style=actor] (77) at (12, -4) {$E$};
		\node [style=none] (78) at (-12, -4) {\textbf{(1)}};
		\node [style=none] (79) at (2, -4) {\textbf{(2)}};
		\node [style=terminatedActor] (80) at (0, -9) {$A$};
		\node [style=none, align=center] (81) at (0, -10.25) {Terminated};
	\end{pgfonlayer}
	\begin{pgfonlayer}{edgelayer}
		\draw [style=reference] (53) to (51);
		\draw [style=reference] (57) to (56);
		\draw [style=reference] (57) to (54);
		\draw [style=message, bend left=330, looseness=0.75] (55) to (54);
		\draw [style=reference] (55) to (57);
		\draw [style=reference] (60.center) to (61.center);
		\draw [style=message] (62.center) to (63.center);
		\draw [style=reference] (73) to (72);
		\draw [style=reference] (77) to (76);
		\draw [style=reference, bend left=345] (77) to (74);
		\draw [style=reference] (75) to (77);
		\draw [style=reference, bend right=15] (74) to (77);
	\end{pgfonlayer}
\end{tikzpicture}
\caption{A simple actor system. Configuration {\bf (1)} leads to the second after $C$ receives the message $m$, which contains a reference to $E$. Notice that an actor can send a message and ``forget'' its reference to the recipient before the message is delivered, as is the case for actor $F$. In both configurations, $E$ is a potential acquaintance of $C$, and $D$ is potentially reachable from $C$. The only terminated actor is $F$ in Configuration {\bf (2)} because all other actors in {\bf (2)} are potentially reachable from unblocked actors.}
\label{fig:actor-graph-example}
\end{figure}

An actor is said to be garbage if it can be destroyed without affecting the system's observable behavior. However, without analyzing an actor’s code, it is not possible to know whether it will have an effect when it receives a message. As is typical for garbage collection algorithms, we will restrict our attention to those actors that can be guaranteed to be garbage without inspecting their internal behavior. According to this more conservative definition, any actor that might receive a message in the future should not be garbage collected because it could, for instance, write to a log file when it becomes busy. Conversely, any actor that remains idle indefinitely can safely be garbage collected because it will never have any effects. We therefore conservatively define an actor to be \emph{terminated} if it is guaranteed to remain idle indefinitely, regardless of how any other actor in the system behaves. For the purposes of this paper, terminated actors coincide with garbage actors.

The definition of termination above is problematic because it requires reasoning about infinite execution traces. Let us now introduce some terminology in order to give an equivalent definition in terms of the \emph{global state} of an actor system \cite{chandyDistributedSnapshotsDetermining1985}.

We say that an actor $B$ is a \emph{potential acquaintance} of $A$ (and $A$ is a \emph{potential inverse acquaintance} of $B$) if $A$ has a reference to $B$ or if there is an undelivered message to $A$ that contains a reference to $B$. We define \emph{potential reachability} to be the reflexive transitive closure of the potential acquaintance relation. That is, $A$ can potentially reach $B$ if there exists a sequence of actors $A_1,\dots,A_n$ where $A = A_1,B=A_n$, and each $A_{i+1}$ is a potential acquaintance of $A_i$. If an actor is idle and has no undelivered messages, then it is \emph{blocked}; otherwise it is \emph{unblocked}. 

Clearly, an unblocked actor is not terminated because it will become busy when the message is delivered. More generally, if an actor is blocked but \emph{potentially reachable} by an unblocked actor, then it is not terminated because it can become unblocked at some point in the future. Consider for example \Cref{fig:actor-graph-example} (1), in which actor $D$ is potentially reachable by actors $F,E,$ and $C$. Once $C$ receives the message in \Cref{fig:actor-graph-example} (2), it becomes a busy actor with a reference to $E$. If $C$ sends a message to $E$ and $E$ sends a message to $D$, then $D$ can become busy.

Conversely, consider an actor $A$ that is not potentially reachable by any unblocked actors, i.e.~$A$ is only potentially reachable by blocked actors. This means that $A$ is only \emph{reachable} by blocked actors, and there are no undelivered messages containing references to any of these actors. Hence there is no way for any of these blocked actors to become unblocked again, so they are all terminated.

Thus, we have shown that an actor is terminated---for a conservative definition of termination---precisely when it is only potentially reachable by blocked actors. One could therefore detect all terminated actors by computing a consistent global snapshot, but computing such snapshots would be infeasible for large actor systems. Instead, we will show how DRL can be used to find the terminated actors in an arbitrary set of \emph{local} actor snapshots $S$. Remarkably, $S$ does not have to be a consistent cut \cite{chandyDistributedSnapshotsDetermining1985}. This is thanks to the state metadata maintained by the DRL communication protocol, which we introduce in \Cref{sec:model}.

\section{Related Work}\label{sec:related-work}

\paragraph*{Global Termination and Snapshots}  
\emph{Global} termination detection (GTD) is used to determine when
\emph{all} processes have terminated
\cite{matternAlgorithmsDistributedTermination1987,matochaTaxonomyDistributedTermination1998}; detecting when individual \emph{actors} have terminated, as we do in this paper, is a more general problem.
For GTD, it suffices to obtain global message send and receive counts.
Most GTD algorithms also assume a fixed process topology. However, Lai
gives an algorithm 
that supports dynamic topologies such as in the actor model~\cite{laiTerminationDetectionDynamically1986}. Lai's
algorithm performs termination detection in ``waves'', disseminating
control messages along a spanning tree (such as an actor supervisor
hierarchy) so as to obtain consistent global message send and receive
counts. Venkatasubramanian et al.~take a similar approach to obtain a
consistent global snapshot of actor states in a distributed
system~\cite{venkatasubramanianScalableDistributedGarbage1992}. However,
such an approach does not scale well because it is not incremental:
garbage cannot be detected until all nodes in the system have
responded.  In contrast, DRL does not require a global snapshot, does not
require actors to coordinate their local snapshots, and does not
require waiting for all nodes before detecting local terminated
actors.

\paragraph*{Reference Tracking} We say that an idle actor is \emph{simple garbage} if it has no undelivered messages and no other actor has a reference to it.
Such actors can be detected with distributed reference counting
\cite{watsonEfficientGarbageCollection1987,bevanDistributedGarbageCollection1987,piquerIndirectReferenceCounting1991}
or with reference listing
\cite{DBLP:conf/iwmm/PlainfosseS95,wangDistributedGarbageCollection2006}
techniques.  In reference listing algorithms, each actor maintains a
partial list of actors that may have references to it. Whenever $A$ sends $B$ a
reference to $C$, it also sends an $\InfoMsg$ message informing $C$
about $B$'s reference. Once $B$ no longer needs a reference to $C$, it
informs $C$ by sending a $\ReleaseMsg$ message; this message should not be processed by $C$ until all preceding messages from $B$ to $C$ have been delivered. Thus an actor is
simple garbage when its reference listing is
empty. 

Our technique uses a form of \emph{deferred reference listing}, in which $A$ may also defer sending $\InfoMsg$
messages to $C$ until it releases its references to $C$.  This allows
$\InfoMsg$ and $\ReleaseMsg$ messages to be batched together, reducing communication
overhead. In a system where most messages contain one or more references, this optimization can reduce the total number of messages by a factor of two.

\paragraph*{Cyclic Garbage}

Actors that are transitively acquainted with one another are said to
form cycles. Cycles of terminated actors are called \emph{cyclic
garbage} and cannot be detected with reference listing alone.  Since
actors are hosted on nodes and cycles may span across multiple nodes,
detecting cyclic garbage requires sharing information between nodes to
obtain a consistent view of the global topology.  One approach is to
compute a global snapshot of the distributed system
\cite{kafuraConcurrentDistributedGarbage1995} using the Chandy-Lamport
algorithm \cite{chandyDistributedSnapshotsDetermining1985}; but this
requires pausing execution of all actors on a node to compute its
local snapshot.

Another approach is to add edges to the actor reference graph so
that actor garbage coincides with passive object garbage
\cite{vardhanUsingPassiveObject2003,wangActorGarbageCollection2010}. This
is convenient because it allows existing algorithms for distributed
passive object GC, such as
\cite{schelvisIncrementalDistributionTimestamp1989}, to be reused in
actor systems. However, such transformations require that actors know
when they have undelivered messages, which necessitates some form of
synchronization. Our approach also adds edges to the actor reference graph, but this time in the form of ``contact tracing'' (\Cref{sec:model}). Unlike prior work, this does not require synchronization and ensures that there is always a ``path'' (instead of a single reference) from an actor to its potential inverse acquaintances (\Cref{sec:chain-lemma}).

To avoid pausing executions, Wang and Varela proposed a reference
listing based technique called the \emph{pseudo-root} algorithm.  The
algorithm computes \emph{approximate} global snapshots and is
implemented in the SALSA runtime
\cite{wangDistributedGarbageCollection2006,wangConservativeSnapshotbasedActor2011}.
The pseudo-root algorithm requires acknowledgments for each application message and each reference contained inside a message, in the worst case producing several times more control messages than application messages. The algorithm also requires actors to write to shared memory (requiring synchronization) if they migrate
or release references during snapshot collection.  Our protocol does not require shared memory and, with batching, only sends control messages when references are released.  Wang and Varela also explicitly address migration of actors, 
a concern orthogonal to our algorithm.

\paragraph{MAC} Our technique was initially inspired by \emph{MAC}, an actor termination detection
algorithm implemented in the Pony runtime
\cite{clebschFullyConcurrentGarbage2013}. In MAC, actors send a local
snapshot to a designated cycle detector whenever their message queue
becomes empty, and send another notification whenever it becomes non-empty. Clebsch and Drossopoulou prove that for systems with
causal message delivery, a simple request-reply protocol is sufficient
to confirm that the cycle detector's view of the topology is
consistent.  However, enforcing causal delivery in a distributed
system imposes additional costs: the well-known strategy of using vector clocks requires $O(n)$ additional memory for each message, where $n$ is the number of nodes \cite{fidge1987timestamps}. Another approach, in which nodes must be arranged in a tree topology, does not impose this overhead but still prevents point-to-point communication between nodes~\cite{blessingTreeTopologiesCausal2017}. DRL is
similar to MAC, but does not require causal message delivery, supports
decentralized termination detection, and does not require actors to take
snapshots each time their message queues become empty. 

We achieve this greater flexibility by reifying information in an actor's local state that would otherwise be implicit in the timing of a snapshot. For instance, all references in DRL are associated with a unique identifying token. If an actor and its acquaintance both include reference $x$ in their snapshot, then the matching token allows us to conclude that the two snapshots are referring to the same reference. Another example of additional information in DRL is the message send and receive counts. In MAC, an actor can infer that it had no undelivered messages at the time of its snapshot if it sends a message to the cycle detector and receives an acknowledgment with no other messages arriving in between. In DRL, we assume that messages can be reordered or delayed by network partitions for an arbitrarily long time. This appears to necessitate the use of message send and receive counts.

\paragraph{Other Versions} An earlier version of DRL appeared in
\cite{plyukhinConcurrentGarbageCollection2018}. In this paper, we
formalize the description of the algorithm and prove its safety and
liveness.  In the process, we discovered that release acknowledgment
messages are unnecessary and that termination detection is more
flexible than we first thought: it is not necessary for GC to be
performed in distinct ``phases'' where every actor takes a snapshot in
each phase.  In particular, once an idle actor takes a snapshot, it
need not take another snapshot until it receives a fresh message.

This paper is an extension of \cite{DBLP:conf/concur/PlyukhinA20}. We have revised the text of several sections for clarity---particularly \Cref{sec:termination-detection}, which now uses the notion of a consistent set of snapshots and an alternative (but equivalent) definition of finalized sets. We also discovered a flaw in our original algorithm for finding maximal finalized sets of snapshots: under certain circumstances, the resulting set may be finalized but not necessarily maximal. The new \Cref{sec:maximal} elaborates on the issues at play and presents an improved algorithm with a proof of correctness. The section also includes the old algorithm as a ``heuristic'' that might be more efficient in some cases. Finally, \Cref{sec:coop} is entirely new and gives an algorithm for efficiently detecting distributed cycles of terminated actors across multiple actor systems.

\section{A Two-Level Semantic Model}\label{sec:model}

In this section we present the DRL communication protocol. We begin by motivating what kind of information a snapshot aggregator would need to detect terminated actors. Next we present the communication protocol more formally, using a two-level semantic model~\cite{venkatasubramanianReasoningMetaLevel1995}. In this model, a \emph{system-level} transition system interprets the operations performed by a user-facing
\emph{application-level} transition system. The application level defines the abstract operational semantics of the actor system from the user's perspective, including location transparency and fairness assumptions. Since DRL preserves this semantics, we leave out the application-level system; for a formalization, see \cite{aghaFoundationActorComputation1997}. The system-level transition system defines each actor's system-level state and what additional actions should be performed when the application level tried to do an operation, such as sending a message or spawning an actor. In the case of DRL, these operations sometimes cause additional system-level messages to be sent or for metadata to be added to an application-level message.

\subsection{Overview}
\label{sec:overview}

Ordinary actor systems allow an actor $A$ to send a message to actor $B$ if $A$ has $B$'s address. In DRL, actors must use \emph{reference objects} (abbreviated \emph{refobs}) instead; refobs combine a plain actor address (the address of the \emph{target}) with additional metadata, such as the address of the refob's designated \emph{owner}. A refob can only be used by its owner: in order for $A$ to give $B$ a reference to $C$, it explicitly creates a new refob owned by $B$. Once a refob is no longer needed, it is \emph{deactivated} by its owner and removed from the owner's local state. These operations could be done manually at the application level or handled automatically in the runtime via a suitable API.

The DRL communication protocol enriches each actor's state with a list of refobs that it currently owns and associated message counts representing the number of messages sent using each refob. Each actor also maintains a subset of the refobs of which it is the target, together with associated message receive counts. Lastly, actors perform a form of ``contact tracing'' by maintaining a subset of the refobs that they have created for other actors; we provide details about the bookkeeping later in this section.

The additional information above allows us to detect termination by inspecting actor snapshots. If a set of snapshots is consistent (in the sense of \cite{chandyDistributedSnapshotsDetermining1985}) then we can use the ``contact tracing'' information to determine whether the set is \emph{closed} under the potential inverse acquaintance relation (see \Cref{sec:chain-lemma}). Then, given a consistent and closed set of snapshots, we can use the message counts to determine whether an actor is blocked. We can therefore find all the terminated actors within a consistent set of snapshots.

In fact, DRL satisfies a stronger property: any set of snapshots that ``appears terminated'' in the sense above is guaranteed to be consistent. Hence, given an arbitrary  closed set of snapshots, it is possible to determine which of the corresponding actors have terminated. This allows a great deal of freedom in how snapshots are aggregated. For instance, each actor could set its own recurring timeout for when to take the next snapshot. The duration of this timeout could be changed based on runtime information, such as how long the actor has been alive; this would amount to a \emph{generational} approach to actor garbage collection \cite{DBLP:journals/cacm/LiebermanH83}.

\paragraph*{Reference Objects}

\begin{figure} \centering \begin{tikzpicture}
	\begin{pgfonlayer}{nodelayer}
		\node [style=actor] (0) at (-11.5, 1) {$C$};
		\node [style=busyActor] (1) at (-6.5, 1) {$A$};
		\node [style=none] (6) at (-9, 1.25) {$y$};
		\node [style=actor] (7) at (3.5, 1) {$C$};
		\node [style=busyActor] (8) at (8.5, 1) {$A$};
		\node [style=actor] (9) at (13.5, 1) {$B$};
		\node [style=none] (10) at (11, 1.25) {$x$};
		\node [style=none] (11) at (6, 1.25) {$y$};
		\node [style=actor] (12) at (-11.5, -2.5) {$C$};
		\node [style=busyActor] (13) at (-6.5, -2.5) {$A$};
		\node [style=actor] (14) at (-1.5, -2.5) {$B$};
		\node [style=none] (15) at (-4, -2.25) {$x$};
		\node [style=none] (16) at (-9, -2.25) {$y$};
		\node [style=none] (17) at (-6.5, -3.75) {$\mathtt{CreatedUsing}(y,z)$};
		\node [style=none] (18) at (-4, -1) {$\mathtt{app}(x,\{z\})$};
		\node [style=actor] (19) at (3.5, -2.5) {$C$};
		\node [style=busyActor] (20) at (8.5, -2.5) {$A$};
		\node [style=busyActor] (21) at (13.5, -2.5) {$B$};
		\node [style=none] (22) at (11.5, -2.25) {$x$};
		\node [style=none] (23) at (6, -2.25) {$y$};
		\node [style=none] (24) at (13.5, -3.75) {$\mathtt{Active}(z)$};
		\node [style=none] (25) at (6, -1) {$\mathtt{info}(y,z)$};
		\node [style=actor] (26) at (-11.5, -6) {$C$};
		\node [style=actor] (27) at (-6.5, -6) {$A$};
		\node [style=busyActor] (28) at (-1.5, -6) {$B$};
		\node [style=none] (29) at (-4, -5.75) {$x$};
		\node [style=none] (30) at (-9, -5.75) {$y$};
		\node [style=none] (31) at (-11.5, -7.25) {$\mathtt{Created}(z)$};
		\node [style=none] (32) at (-6.5, -8) {$\mathtt{release}(z)$};
		\node [style=actor] (33) at (3.5, -6) {$C$};
		\node [style=actor] (34) at (8.5, -6) {$A$};
		\node [style=actor] (35) at (13.5, -6) {$B$};
		\node [style=none] (36) at (11, -5.75) {$x$};
		\node [style=none] (37) at (6, -5.75) {$y$};
		\node [style=none] (38) at (3.5, -7.25) {$\mathtt{Created}(z)$};
		\node [style=none] (39) at (3.5, -8) {$\mathtt{Released}(z)$};
		\node [style=busyActor] (40) at (-2.5, -10) {$A$};
		\node [style=none] (42) at (1.75, -10) {};
		\node [style=none] (43) at (3.25, -10) {};
		\node [style=none] (44) at (6.75, -10) {};
		\node [style=none] (45) at (8.25, -10) {};
		\node [style=none] (49) at (2.5, -9.75) {$x$};
		\node [style=none] (50) at (7.5, -9.75) {$m$};
		\node [style=none] (51) at (1.5, 1) {\textbf{(2)}};
		\node [style=none] (52) at (1.5, -2.5) {\textbf{(4)}};
		\node [style=none] (53) at (1.5, -6) {\textbf{(6)}};
		\node [style=none] (54) at (-13.5, 1) {\textbf{(1)}};
		\node [style=none] (55) at (-13.5, -2.5) {\textbf{(3)}};
		\node [style=none] (56) at (-13.5, -6) {\textbf{(5)}};
		\node [style=none] (57) at (8.5, -4.25) {$z$};
		\node [style=actor] (58) at (-7.5, -10) {$A$};
		\node [style=none] (60) at (-2.5, -11.325) {Busy Actor};
		\node [style=none] (61) at (2.5, -11.25) {};
		\node [style=none] (62) at (2.5, -11.275) {Reference};
		\node [style=none] (63) at (7.5, -11.35) {Message};
		\node [style=none] (64) at (-7.5, -11.25) {Idle Actor};
	\end{pgfonlayer}
	\begin{pgfonlayer}{edgelayer}
		\draw [style=reference] (1) to (0);
		\draw [style=reference, in=180, out=0] (8) to (9);
		\draw [style=reference] (8) to (7);
		\draw [style=reference, in=180, out=0] (13) to (14);
		\draw [style=message, in=150, out=30] (13) to (14);
		\draw [style=reference] (13) to (12);
		\draw [style=reference, in=180, out=0] (20) to (21);
		\draw [style=reference] (20) to (19);
		\draw [style=message, bend left=330] (20) to (19);
		\draw [style=reference, in=180, out=0] (27) to (28);
		\draw [style=reference] (27) to (26);
		\draw [style=message, bend left, looseness=0.75] (28) to (26);
		\draw [style=reference, in=180, out=0] (34) to (35);
		\draw [style=reference] (34) to (33);
		\draw [style=reference] (42.center) to (43.center);
		\draw [style=message] (44.center) to (45.center);
		\draw [style=reference, bend left, looseness=0.75] (21) to (19);
	\end{pgfonlayer}
\end{tikzpicture}
    \caption{An example showing how refobs are created and
destroyed. Below each actor we list all the ``facts'' related to $z$ that are stored in its local state. Although not pictured in the figure, $A$ also obtains facts $\Activated(x)$ and $\Activated(y)$ after spawning actors $B$ and $C$, respectively. Likewise, actors $B,C$ obtain facts $\Created(x),\Created(y)$, respectively, upon being spawned.}
    \label{fig:refob-example}
\end{figure}

A refob is a triple $(x,A,B)$, where $A$ is the owner actor's address, $B$ is the target actor's address, and $x$ is a globally unique token. An actor can cheaply generate such a token by combining its address with a local sequence number, since actor systems already guarantee that each address is unique. We will stylize a triple $(x,A,B)$ as $\Refob x A B$. We will also sometimes refer to such a refob as simply $x$, since tokens act as unique identifiers.

When an actor $A$ spawns an actor $B$ (\Cref{fig:refob-example}
(1, 2)) the DRL protocol creates a new refob
$\Refob x A B$ that is stored in both $A$ and $B$'s system-level
state, and a refob $\Refob w B B$ in $B$'s state. The refob $x$ allows $A$ to send application-level messages to
$B$. These messages are denoted $\AppMsg(x,R)$, where $R$ is the set of refobs contained in the message that $A$ has created for $B$. The refob $y$ corresponds to the \texttt{self} variable present in some actor languages. 

If $A$ has active refobs $\Refob x A B$ and $\Refob y A C$, then it can
create a new refob $\Refob z B C$ by generating a token $z$. In
addition to being sent to $B$, this refob must also temporarily be
stored in $A$'s system-level state and marked as ``created using $y$''
(\Cref{fig:refob-example} (3)). Once $B$ receives $z$, it must add
the refob to its system-level state and mark it as ``active''
(\Cref{fig:refob-example} (4)). Note that $B$ can have multiple
distinct refobs that reference the same actor in its state; this
can be the result of, for example, several actors concurrently sending
refobs to $B$.  Transition rules for spawning actors and sending
messages are given in \Cref{sec:standard-actor-operations}.

Actor $A$ may remove $z$ from its state once it has sent a
(system-level) $\InfoMsg$ message informing $C$ about $z$
(\Cref{fig:refob-example} (4)). Similarly, when $B$ no longer
needs its refob for $C$, it can ``deactivate'' $z$ by removing it
from local state and sending $C$ a (system-level) $\ReleaseMsg$
message (\Cref{fig:refob-example} (5)). Note that if $B$ already
has a refob $\Refob z B C$ and then receives another $\Refob {z'} B C$,
then it can be more efficient to defer deactivating the extraneous
$z'$ until $z$ is also no longer needed; this way, the $\ReleaseMsg$
messages can be batched together.

When $C$ receives an $\InfoMsg$ message, it records that the refob
has been created, and when $C$ receives a $\ReleaseMsg$ message, it
records that the refob has been released
(\Cref{fig:refob-example} (6)).  Note that these messages may
arrive in any order. Once $C$ has received both, it is permitted to
remove all facts about the refob from its local state. Transition
rules for these reference listing actions are given in
\Cref{sec:release-protocol}.

Once a refob has been created, it cycles through four states:
pending, active, inactive, or released.  A refob $\Refob z B C$ is
said to be \emph{pending} until it is received by its owner $B$. Once
received, the refob is \emph{active} until it is \emph{deactivated}
by its owner, at which point it becomes \emph{inactive}.  Finally,
once $C$ learns that $z$ has been deactivated, the refob is said to
be \emph{released}.  A refob that has not yet been released is
\emph{unreleased}.  

Slightly amending the definition we gave in \Cref{sec:background}, we say that $B$ is a \emph{potential acquaintance} of $A$
(and $A$ is a \emph{potential inverse acquaintance} of $B$) when there
exists an unreleased refob $\Refob x A B$. Thus, $B$ becomes a potential acquaintance of $A$ as soon as $x$ is created, and only ceases to be an acquaintance once it has received a $\ReleaseMsg$ message for every refob $\Refob y A B$ that has been created so far.

\begin{figure} \centering \begin{tikzpicture}
	\begin{pgfonlayer}{nodelayer}
		\node [style=none, rotate=-45] (70) at (6.85, -2) {\small $\mathtt{release}(z)$};
		\node [style=none, rotate=-45] (81) at (7, 2.45) {\small $\mathtt{app}(x',\{z'\})$};
		\node [style=none, rotate=-45] (91) at (12, -2) {\small $\mathtt{app}(z')$};
		\node [style=none] (98) at (-11.5, 5) {};
		\node [style=none] (99) at (-11.5, 5) {$A$};
		\node [style=none] (100) at (-10.75, 5) {};
		\node [style=none] (101) at (18, 5) {};
		\node [style=none] (102) at (-10.75, 0) {};
		\node [style=none] (103) at (18, 0) {};
		\node [style=none] (104) at (-10.75, -5) {};
		\node [style=none] (105) at (18, -5) {};
		\node [style=none] (106) at (-11.5, 0) {$B$};
		\node [style=none] (107) at (-11.5, -5) {$C$};
		\node [style=none] (108) at (-8.25, 0) {};
		\node [style=none] (109) at (0.75, -5) {};
		\node [style=none] (110) at (4.25, 0) {};
		\node [style=none] (111) at (9.25, -5) {};
		\node [style=none] (112) at (9.25, 0) {};
		\node [style=none] (113) at (14.25, -5) {};
		\node [style=none] (114) at (3.5, 5) {};
		\node [style=none] (115) at (8.5, 0) {};
		\node [style=none] (116) at (0.25, 0.25) {};
		\node [style=none] (117) at (3.25, -5.5) {};
		\node [style=none] (122) at (2.25, -2.5) {$t_1$};
		\node [style=none, rotate=-30] (126) at (-6.25, -1.75) {\small $\mathtt{app}(z)$};
		\node [style=none] (127) at (7, 0.25) {};
		\node [style=none] (128) at (11.25, -5.5) {};
		\node [style=none] (129) at (16.5, 0.25) {};
		\node [style=none] (130) at (16.5, -5.5) {};
		\node [style=none] (131) at (9.25, -2.25) {$t_2$};
		\node [style=none] (132) at (17, -2.25) {$t_3$};
		\node [style=none] (133) at (-9, 0.25) {};
		\node [style=none] (134) at (-9, -5.5) {};
		\node [style=none] (135) at (-8.5, -2.5) {$t_0$};
		\node [style=none] (136) at (-3.75, 0) {};
		\node [style=none] (137) at (-1.75, -5) {};
		\node [style=none, rotate=-68.2] (138) at (-2.5, -1.75) {\small $\mathtt{info}(z,y)$};
		\node [style=none] (139) at (-5.25, 5) {};
		\node [style=none] (140) at (-7.5, 0) {};
		\node [style=none, rotate=68] (141) at (-6.75, 3) {\small $\mathtt{app}(x,\{y\})$};
		\node [style=none, rotate=-45] (142) at (1.75, 2.45) {\small $\mathtt{app}(x')$};
		\node [style=none] (143) at (-1.75, 5) {};
		\node [style=none] (144) at (3.25, 0) {};
		\node [style=none] (145) at (-9, -6) {\small $\emptyset$};
		\node [style=none] (146) at (-9, 0.75) {\small $\emptyset$};
		\node [style=none] (147) at (0.25, 0.75) {\small $\mathtt{Sent}(z,2)$};
		\node [style=none] (148) at (3.25, -6) {\small $\mathtt{Recv}(z,2)$};
		\node [style=none] (149) at (11.25, -6) {\small $\emptyset$};
		\node [style=none] (150) at (7, 0.75) {\small $\emptyset$};
		\node [style=none] (151) at (16.5, -6) {\small $\mathtt{Recv}(z',2)$};
		\node [style=none] (152) at (16.5, 0.75) {\small $\mathtt{Sent}(z',2)$};
		\node [style=none] (153) at (-4.75, 0.25) {};
		\node [style=none] (154) at (-4.75, -0.25) {};
		\node [style=none] (155) at (-0.75, -4.75) {};
		\node [style=none] (156) at (-0.75, -5.5) {};
		\node [style=none] (157) at (-4.75, 0.75) {\small $\mathtt{Sent}(z,1)$};
		\node [style=none] (158) at (-0.75, -6) {\small $\mathtt{Recv}(z,1)$};
		\node [style=none, rotate=-45] (159) at (13.75, -2) {\small $\mathtt{app}(z')$};
		\node [style=none] (160) at (11, 0) {};
		\node [style=none] (161) at (16, -5) {};
	\end{pgfonlayer}
	\begin{pgfonlayer}{edgelayer}
		\draw [style=reference] (100.center) to (101.center);
		\draw [style=reference] (102.center) to (103.center);
		\draw [style=reference] (104.center) to (105.center);
		\draw [style=message] (108.center) to (109.center);
		\draw [style=message] (110.center) to (111.center);
		\draw [style=message] (112.center) to (113.center);
		\draw [style=message] (114.center) to (115.center);
		\draw [style=time horizon, in=90, out=-90, looseness=1.25] (116.center) to (117.center);
		\draw [style=time horizon, in=90, out=-90, looseness=1.25] (127.center) to (128.center);
		\draw [style=time horizon, in=90, out=-90, looseness=1.25] (129.center) to (130.center);
		\draw [style=time horizon, in=90, out=-90, looseness=1.25] (133.center) to (134.center);
		\draw [style=message] (136.center) to (137.center);
		\draw [style=message] (140.center) to (139.center);
		\draw [style=message] (143.center) to (144.center);
		\draw [style=time horizon] (153.center) to (154.center);
		\draw [style=time horizon] (155.center) to (156.center);
		\draw [style=message] (160.center) to (161.center);
	\end{pgfonlayer}
\end{tikzpicture}
    \caption{An event diagram~\cite{books/daglib/0066897} for actors $A,B,C$, illustrating message counts and consistent snapshots. Dashed arrows represent messages and dotted lines represent mutually quiescent cuts. For a cut to be mutually quiescent, it is necessary (but not sufficient) that the message send and receive counts agree for all participants.}
    \label{fig:message-counts}
\end{figure}

\paragraph*{Message Counts and Snapshots}

For each refob $\Refob x A B$, the owner $A$ maintains a count of how many $\AppMsg$ and $\InfoMsg$ messages have been sent along $x$; this count can be deleted when $A$
deactivates $x$. Each message is annotated with the refob used to
send it. Whenever $B$ receives an $\AppMsg$ or $\InfoMsg$ message along $x$, it
correspondingly increments a receive count for $x$; this count can be deleted once $x$
has been released. Thus the memory overhead of message counts is linear in
the number of unreleased refobs. \Cref{fig:message-counts} gives an example.

A snapshot is a copy of all the facts in an actor's system-level state at some point in time. We will assume throughout the paper that in every set of snapshots $Q$, each snapshot was taken by a different actor; such a set is also said to form a \emph{cut}. Recall that a set of snapshots $Q$ is \emph{consistent} if no snapshot in $Q$ causally precedes any other \cite{chandyDistributedSnapshotsDetermining1985}; it is as if all the actors in $Q$ took their snapshots simultaneously. Let us also say that $Q$ is \emph{mutually quiescent} if for all actors $A,B$ in $Q$, all messages sent from $A$ to $B$ before $A$'s snapshot were also received before $B$'s snapshot.  Notice that mutual quiescence is just a special case of consistency, in which all messages sent by actors in the cut to actors in the cut have been delivered. Moreover, if each actor in $Q$ is idle and $Q$ contains each actor's potential inverse acquaintances, then $Q$ corresponds to a terminated set of actors: every actor in $Q$ is blocked and only potentially reachable by other blocked actors in $Q$.

One might therefore hope to check whether $Q$ is mutually quiescent by simply comparing the message send and receive counts of all snapshots in $Q$. Clearly, if $Q$ is mutually quiescent, then the participants' send and receive counts will agree for each $\Refob x A B$ where $A,B \in Q$. However, the converse may be false for two reasons: out of order delivery of messages, and temporarily null message counts.

\begin{enumerate}
    \item \emph{Out of order delivery:} In \Cref{fig:message-counts}, a snapshot from $B$ when $\SentCount(z,1)$ is in its knowledge set would not be consistent with a snapshot from $C$ when $\RecvCount(z,1)$ is in its knowledge set. This is because the message $\InfoMsg(z,y)$ is sent after $B$'s snapshot and received before $C$'s snapshot. To guarantee this situation does not occur, we must be able to prove that $B$ does not send any messages along $z$ in the interval between $B$'s snapshot and $C$'s snapshot. In particular, this holds when $C$'s snapshot happens before $B$'s snapshot. 
    \item \emph{Null message count:} Based on the available information, $C$'s snapshot at $t_0$ in \Cref{fig:message-counts} appears mutually quiescent with $B$'s snapshot at $t_2$ and $C$'s snapshot at $t_2$ appears mutually quiescent with $B$'s snapshot at $t_0$---despite neither of these pairs being truly mutually quiescent. The problem is that when $B$'s send count for $\Refob z B C$ is null, it could be because $B$ has not yet received $z$ or because $B$ has already deactivated $z$. Likewise, $C$'s receive count could be null because it has not yet received any messages along $z$ or because $z$ has already been released. 
\end{enumerate}
To distinguish these scenarios, we incorporate the snapshots of $C$'s other potential inverse acquaintances---such as $A$---into the snapshot set $Q$. In \Cref{sec:chain-lemma} we identify a distributed property called the \emph{Chain Lemma} that must hold in any consistent set of snapshots closed under the potential inverse acquaintance relation. We show, in \Cref{sec:termination-detection}, that combining the Chain Lemma with message counts is sufficient to determine whether a set of snapshots is mutually quiescent.

\paragraph*{Definitions}

We use the capital letters $A,B,C,D,E$ to denote actor addresses.
Tokens are denoted $x,y,z$, with a special reserved token $\NullToken$
for messages from external actors.

A \emph{fact} is a value that takes one of the following forms:
$\Created(x)$, $\Released(x)$, $\CreatedUsing(x,y)$, $\Activated(x)$, $\Unreleased(x)$,
$\SentCount(x,n)$, or $\RecvCount(x,n)$ for some refobs $x,y$ and
natural number $n$.  Each actor's state holds a set of facts about
refobs and message counts called its \emph{knowledge set}.  We use
$\phi,\psi$ to denote facts and $\Phi,\Psi$ to denote finite sets of
facts.  Each fact may be interpreted as a \emph{predicate} that
indicates the occurrence of some past event. Interpreting a set of
facts $\Phi$ as a set of axioms, we write $\Phi \vdash \phi$ when
$\phi$ is derivable by first-order logic from $\Phi$ with the
following additional rules: 
\begin{itemize}
    \item If $(\not\exists n \in \mathbb N,\ \SentCount(x,n) \in
\Phi)$ then $\Phi \vdash \SentCount(x,0)$
    \item If $(\not\exists n \in \mathbb N,\ \RecvCount(x,n) \in
\Phi)$ then $\Phi \vdash \RecvCount(x,0)$
    \item If $\Phi \vdash \Created(x) \land \lnot \Released(x)$ then
$\Phi \vdash \Unreleased(x)$
    \item If $\Phi \vdash \CreatedUsing(x,y)$ then $\Phi \vdash
\Created(y)$
\end{itemize}
For convenience, we define a pair of functions
$\IncSent(x,\Phi),\IncRecv(x,\Phi)$ for incrementing message
send/receive counts, as follows: If $\SentCount(x,n) \in \Phi$ for some
$n$, then
$\IncSent(x,\Phi) = (\Phi \setminus \{\SentCount(x,n)\}) \cup
\{\SentCount(x,n+1)\}$; otherwise,
$\IncSent(x,\Phi) = \Phi \cup \{\SentCount(x,1)\}$. Likewise for
$\IncRecv$ and $\RecvCount$.

Recall that an actor is either \emph{busy} (processing a message) or
\emph{idle} (waiting for a message). An actor with knowledge set
$\Phi$ is denoted $[\Phi]$ if it is busy and $(\Phi)$ if it is idle.

Our specification includes both \emph{system messages} (also called
\emph{control messages}) and \emph{application messages}. The former
are automatically generated by the DRL  protocol and handled at the system
level, whereas the latter are explicitly created and consumed by
user-defined behaviors. Application-level messages are denoted
$\AppMsg(x,R)$. The argument $x$ is the refob used to send the
message. The second argument $R$ is a set of refobs created by the
sender to be used by the destination actor. Any remaining application-specific data in the message is omitted in our notation.

The DRL communication protocol uses two kinds of system messages. $\InfoMsg(y, z, B)$ is a message sent from an actor $A$ to an actor $C$, informing it that a new refob $\Refob z B C$ was created using $\Refob y A C$. $\ReleaseMsg(x,n)$ is a message sent from an actor $A$ to an actor $B$, informing it that the refob $\Refob x A B$ has been deactivated and that a total of $n$ messages have been sent along $x$.

A \emph{configuration} $\Config{\alpha}{\mu}{\rho}{\chi}$ is a
quadruple $(\alpha,\mu,\rho,\chi)$ where: $\alpha$ is a mapping from actor addresses to knowledge sets; $\mu$ is a mapping from actor addresses to multisets of messages; and $\rho,\chi$ are sets of actor addresses. Actors in $\dom(\alpha)$ are \emph{internal actors} and actors in $\chi$ are
\emph{external actors}; the two sets may not intersect. The mapping $\mu$ associates each actor with undelivered messages to that actor. Actors in
$\rho$ are \emph{receptionists}.  We will ensure $\rho \subseteq \dom(\alpha)$ remains
valid in any configuration that is derived from a configuration where
the property holds (referred to as the locality laws in
\cite{Baker-Hewitt-laws77}).

Configurations are denoted by $\kappa$, $\kappa'$, $\kappa_0$,
etc. If an actor address $A$ (resp. a token $x$), does not occur in
$\kappa$, then the address (resp. the token) is said to be
\emph{fresh}.  We assume a facility for generating fresh addresses and
tokens.

In order to express our transition rules in a pattern-matching style, we will employ the following shorthand. Let $\alpha,[\Phi]_A$ refer to a
mapping $\alpha'$ where $\alpha'(A) = [\Phi]$ and $\alpha =
\alpha'|_{\dom(\alpha') \setminus \{A\}}$.  Similarly, let
$\mu,\Msg{A}{m}$ refer to a mapping $\mu'$ where $m \in \mu'(A)$ and
$\mu = \mu'|_{\dom(\mu') \setminus \{A\}} \cup \{A \mapsto \mu'(A)
\setminus \{m\}\}$. Informally, the expression $\alpha,[\Phi]_A$ refers to a set of actors containing both $\alpha$ and the busy actor $A$ (with knowledge set $\Phi$); the expression $\mu, \Msg{A}{m}$ refers to the set of messages containing both $\mu$ and the message $m$ (sent to actor $A$).

The rules of our transition system define atomic transitions from one configuration
to another.  Each transition rule has a label $l$, parameterized by some
variables $\vec x$ that occur in the left- and right-hand
configurations. Given a configuration $\kappa$, these parameters
functionally determine the next configuration $\kappa'$. Given
arguments $\vec v$, we write $\kappa \Step{l(\vec v)} \kappa'$ to denote a semantic step from $\kappa$ to $\kappa'$ using rule $l(\vec v)$.

We refer to a label with arguments $l(\vec v)$ as an \emph{event},
denoted $e$. A sequence of events is denoted $\pi$. If $\pi =
e_1,\dots,e_n$ then we write $\kappa \Step \pi \kappa'$ when $\kappa
\Step{e_1} \kappa_1 \Step{e_2} \dots \Step{e_n} \kappa'$. If there
exists $\pi$ such that $\kappa \Step \pi \kappa'$, then $\kappa'$ is
\emph{derivable} from $\kappa$. An \emph{execution} (also called a \emph{computation path} \cite{aghaFoundationActorComputation1997}) is a sequence of events $e_1,\dots,e_n$ such that
$\kappa_0 \Step{e_1} \kappa_1 \Step{e_2} \dots \Step{e_n} \kappa_n$,
where $\kappa_0$ is the initial configuration
(\Cref{sec:initial-configuration}).  We say that a property holds \emph{at time $t$} if it holds in $\kappa_t$. We will also employ the shorthand that $\alpha_t$ is the actor configuration at time $t$, i.e. $\kappa_t = \Config{\alpha_t}{\mu}{\rho}{\chi}$.

Note that the DRL communication protocol does not require a notion of a unique global time. We could have given a more general specification using concurrent rewriting~\cite{DBLP:journals/tcs/Meseguer92}, in which potential executions are partial orders of events. A given execution in such a specification can be mapped to an execution in our system by mapping its partial order to a total order which respects the ordering specified in the partial order. We refer to ``time'' as an ordinal corresponding to an arbitrary total order that is consistent with a partial order in a system's execution (see \cite{10.5555/889486,books/daglib/0066897}). This allows us to prove various properties by induction on time $t$ instead of by more complicated means.

\subsection{Initial Configuration}\label{sec:initial-configuration}

The initial configuration $\kappa_0$ consists of a single actor in a
busy state:
$$\Config{[\Phi]_A}{\emptyset}{\emptyset}{\{E\}},$$
where
$\Phi = \{\Activated(\Refob x A E),\ \Created(\Refob y A A),\
\Activated(\Refob y A A)\}$. The actor's knowledge set includes a
refob to itself and a refob to an external actor $E$. $A$ can
become a receptionist by sending $E$ a refob to itself.
Henceforth, we will only consider configurations that are derivable
from an initial configuration.

\subsection{Standard Actor Operations}\label{sec:standard-actor-operations}

\begin{figure}[t]
$\textsc{Spawn}(x, A, B)$ 
$$\Config{\alpha, [\Phi]_A}{\mu}{\rho}{\chi} \InternalStep \Config{\alpha, [\Phi \cup \{ \Activated(\Refob x A B) \}]_A, [\Psi]_B}{\mu}{\rho}{\chi}$$
\begin{tabular}{ll}
where & $x,y,B$  fresh\\
and & $\Psi = \{ \Created(\Refob x A B),\ \Created(\Refob {y} B B),\ \Activated(\Refob y B B) \}$
\end{tabular}

\vspace{0.5cm}

$\textsc{Send}(x,\vec y, \vec z, A, B,\vec C)$ 
$$\Config{\alpha, [\Phi]_A}{\mu}{\rho}{\chi} \InternalStep \Config{\alpha, [\IncSent(x,\Phi) \cup \Psi]_A}{\mu, \Msg{B}{\AppMsg(x,R)}}{\rho}{\chi}$$
\begin{tabular}{ll}
where & $\vec z$ fresh and $n = |\vec y| = |\vec z| = |\vec C|$\\
and & $\Phi \vdash \Activated(\Refob x A B)$ and $\forall i \le n,\ \Phi \vdash \Activated(\Refob{y_i}{A}{C_i})$\\
and & $R = \{\Refob{z_i}{B}{C_i}\ |\ i \le n \}$ and $\Psi = \{\CreatedUsing(y_i,z_i)\ |\ i \le n \}$
\end{tabular}

\vspace{0.5cm}

$\textsc{Receive}(x,B,R)$ 
$$\Config{\alpha, (\Phi)_B}{\mu, \Msg{B}{\AppMsg(x,R)}}{\rho}{\chi} \InternalStep \Config{\alpha, [\IncRecv(x,\Phi) \cup \Psi]_B}{\mu}{\rho}{\chi}$$
\begin{tabular}{ll}
where $\Psi = \{\Activated(z)\ |\ z \in R\}$
\end{tabular}

\vspace{0.5cm}

$\textsc{Idle}(A)$ 
$$\Config{\alpha, [\Phi]_A}{\mu}{\rho}{\chi} \InternalStep \Config{\alpha, (\Phi)_A}{\mu}{\rho}{\chi}$$

  \caption{Rules for standard actor interactions.}
  \label{rules:actors}
\end{figure}

\Cref{rules:actors} gives transition rules for standard actor operations, such as spawning actors and sending messages. Each of these rules corresponds a rule in the standard operational semantics of actors~\cite{aghaFoundationActorComputation1997}. Note that each rule is atomic, but can just as well be implemented as a sequence of several smaller steps without loss of generality because actors do not share state---see \cite{aghaFoundationActorComputation1997} for a formal proof.

The \textsc{Spawn} event allows a busy actor $A$ to spawn a new actor $B$ and creates two refobs $\Refob x A B,\ \Refob y B B$. $B$ is initialized with knowledge about $x$ and $y$ via the facts $\Created(x),\Created(y)$. The facts $\Activated(x), \Activated(y)$ allow $A$ and $B$ to immediately begin sending messages to $B$. Note that implementing \textsc{Spawn} does not require a synchronization protocol between $A$ and $B$ to construct $\Refob x A B$. The parent $A$ can pass both its address and the freshly generated token $x$ to the constructor for $B$. Since actors typically know their own addresses, this allows $B$ to construct the triple $(x,A,B)$. Since the \texttt{spawn} call typically returns the address of the spawned actor, $A$ can also create the same triple.

The \textsc{Send} event allows a busy actor $A$ to send an application-level message to $B$ containing a set of refobs $z_1,\dots,z_n$ to actors $\vec C = C_1,\dots,C_n$.  Note that it is possible that $B = A$ or $C_i = A$ for some $i$ in this sequence---i.e., an actor may send itself a message, or it may send $B$ its a refob containing its own address.  For each new refob $z_i$, we say that the message \emph{contains $z_i$}. Any other data in the message besides these refobs is irrelevant to termination detection and therefore omitted. To send the message, $A$ must have active refobs to both the target actor $B$ and to every actor $C_1,\dots,C_n$ referenced in the message. For each target $C_i$, $A$ adds a fact $\CreatedUsing(y_i,z_i)$ to its knowledge set; we say that $A$ \emph{created $z_i$ using $y_i$}. Finally, $A$ must increment its $\SentCount$ count for the refob $x$ used to send the message; we say that the message is sent \emph{along $x$}.

The \textsc{Receive} event allows an idle actor $B$ to become busy by consuming an application message sent to $B$. Before  performing subsequent actions, $B$ increments the receive count for $x$ and adds all refobs in the message to its knowledge set.

Finally, the \textsc{Idle} event puts a busy actor into the idle state, enabling it to consume another message.

\subsection{Release Protocol}\label{sec:release-protocol}

\begin{figure}[t!]

$\textsc{SendInfo}(y,z,A,B,C)$ 
$$\Config{\alpha, [\Phi \cup \Psi]_A}{\mu}{\rho}{\chi} \InternalStep \Config{\alpha, [\IncSent(y,\Phi)]_A}{\mu,\Msg{C}{\InfoMsg(y,z,B)}}{\rho}{\chi}$$
\begin{tabular}{ll}
where $\Psi = \{\CreatedUsing(\Refob y A C,\Refob z B C)\}$
\end{tabular}

\vspace{0.5cm}

$\textsc{Info}(y,z,B,C)$ 
$$\Config{\alpha, (\Phi)_C}{\mu, \Msg{C}{\InfoMsg(y,z,B)}}{\rho}{\chi} \InternalStep \Config{\alpha, (\IncRecv(y,\Phi) \cup \Psi)_C}{\mu}{\rho}{\chi}$$
\begin{tabular}{ll}
where $\Psi = \{\Created(\Refob z B C)\}$
\end{tabular}

\vspace{0.5cm}

$\textsc{SendRelease}(x,A,B)$ 
$$\Config{\alpha, [\Phi \cup \Psi]_A}{\mu}{\rho}{\chi} \InternalStep \Config{\alpha, [\Phi]_A}{\mu, \Msg{B}{\ReleaseMsg(x,n)}}{\rho}{\chi}$$
\begin{tabular}{ll}
where &$\Psi = \{\Activated(\Refob x A B), \SentCount(x,n)\}$\\
and & $\not\exists y,\ \CreatedUsing(x,y) \in \Phi$
\end{tabular}

\vspace{0.5cm}

$\textsc{Release}(x,A,B)$
$$\Config{\alpha, (\Phi)_B}{\mu, \Msg{B}{\ReleaseMsg(x,n)}}{\rho}{\chi} \InternalStep \Config{\alpha, (\Phi \cup \{\Released(x)\})_B}{\mu}{\rho}{\chi}$$
\begin{tabular}{l}
only if $\Phi \vdash \RecvCount(x,n)$
\end{tabular}

\vspace{0.5cm}

$\textsc{Compaction}(x,B,C)$ 
$$\Config{\alpha, (\Phi \cup \Psi)_C}{\mu}{\rho}{\chi} \InternalStep \Config{\alpha, (\Phi)_C}{\mu}{\rho}{\chi}$$
\begin{tabular}{ll}
where & $\Psi = \{\Created(\Refob x B C), \Released(\Refob x B C), \RecvCount(x,n)\}$ \\&for some $n \in \mathbb N$\\
or & $\Psi = \{\Created(\Refob x B C), \Released(\Refob x B C)\}$ and\\
& $\forall n \in \mathbb N,\ \RecvCount(x,n) \not\in \Phi$
\end{tabular}

\vspace{0.5cm}

$\textsc{Snapshot}(A, \Phi)$ 
$$\Config{\alpha, (\Phi)_A}{\mu}{\rho}{\chi} \InternalStep \Config{\alpha, (\Phi)_A}{\mu}{\rho}{\chi}$$

  \caption{Rules for performing the release protocol.}
  \label{rules:release}
\end{figure}

Whenever an actor creates or receives a refob, it adds facts to its knowledge set. To remove these facts when they are no longer needed, actors can perform the \emph{release protocol} defined in \Cref{rules:release}. All of these rules are not present in the standard operational semantics of actors.

The \textsc{SendInfo} event allows a busy actor $A$ to inform $C$ about a refob $\Refob z B C$ that it created using $y$; we say that the $\InfoMsg$ message is sent \emph{along $y$} and \emph{contains $z$}. This event allows $A$ to remove the fact $\CreatedUsing(y,z)$ from its knowledge set. It is crucial that $A$ also increments its $\SentCount$ count for $y$ to indicate an undelivered $\InfoMsg$ message sent to $C$: it allows the snapshot aggregator to detect when there are undelivered $\InfoMsg$ messages, which contain refobs. This message is delivered with the \textsc{Info} event, which adds the fact $\Created(\Refob z B C)$ to $C$'s knowledge set and correspondingly increments $C$'s $\RecvCount$ count for $y$.

When an actor $A$ no longer needs $\Refob x A B$ for sending messages, $A$ can deactivate $x$ with the \textsc{SendRelease} event; we say that the $\ReleaseMsg$ is sent \emph{along $x$}. A precondition of this event is that $A$ has already sent messages to inform $B$ about all the refobs it has created using $x$. In practice, an implementation may defer sending any $\InfoMsg$ or $\ReleaseMsg$ messages to a target $B$ until all $A$'s refobs to $B$ are deactivated. This introduces a trade-off between the number of control messages and the rate of simple garbage detection (\Cref{sec:chain-lemma}).

Each $\ReleaseMsg$ message for a refob $x$ includes a count $n$ of the number of messages sent using $x$. This ensures that $\ReleaseMsg(x,n)$ is only delivered after all the preceding messages sent along $x$ have been delivered. Once the \textsc{Release} event can be executed, it adds the fact that $x$ has been released to $B$'s knowledge set. Once $C$ has received both an $\InfoMsg$ and $\ReleaseMsg$ message for a refob $x$, it may remove facts about $x$ from its knowledge set using the \textsc{Compaction} event.

Finally, the \textsc{Snapshot} event captures an idle actor's knowledge set. For simplicity, we have omitted the process of disseminating snapshots to an aggregator. Although this event does not change the configuration, it allows us to prove properties about snapshot events at different points in time.

\subsection{Composition and Effects}\label{sec:actor-composition}

\begin{figure}
$\textsc{In}(x,A,R)$ 
$$\Config{\alpha}{\mu}{\rho}{\chi} \ExternalStep \Config{\alpha}{\mu, \Msg{A}{\AppMsg(x, R)}}{\rho}{\chi \cup \chi'}$$
\begin{tabular}{ll}
where & $A \in \rho$ and $R = \{ \Refob{x_1}{A}{B_1}, \dots, \Refob{x_n}{A}{B_n} \}$ and $x_1,\dots,x_n$ fresh\\
and & $\{B_1,\dots,B_n\} \cap \dom(\alpha) \subseteq \rho$ and $\chi' = \{B_1,\dots,B_n\} \setminus \dom(\alpha)$ \\
\end{tabular}

\vspace{0.5cm}

$\textsc{Out}(x,B,R)$
$$\Config{\alpha}{\mu,\Msg{B}{\AppMsg(x, R)}}{\rho}{\chi} \ExternalStep \Config{\alpha}{\mu}{\rho \cup \rho'}{\chi}$$
\begin{tabular}{ll}
where $B \in \chi$ and $R = \{ \Refob{x_1}{B}{C_1}, \dots, \Refob{x_n}{B}{C_n} \}$ and $\rho' = \{C_1,\dots,C_n\} \cap \dom(\alpha)$
\end{tabular}

\vspace{0.5cm}

$\textsc{ReleaseOut}(x,B)$ 
$$\Config{\alpha}{\mu,\Msg{B}{\ReleaseMsg(x,n)}}{\rho}{\chi \cup \{B\}} \ExternalStep \Config{\alpha}{\mu}{\rho}{\chi \cup \{B\}}$$

\vspace{0.2cm}

$\textsc{InfoOut}(y,z,A,B,C)$ 
$$\Config{\alpha}{\mu,\Msg{C}{\InfoMsg(y,z,A,B)}}{\rho}{\chi \cup \{C\}} \ExternalStep \Config{\alpha}{\mu}{\rho}{\chi \cup \{C\}}$$

  \caption{Rules for interacting with the outside world.}
  \label{rules:composition}
\end{figure}

We give rules to dictate how internal actors interact with external actors in
\Cref{rules:composition}. The \textsc{In} and \textsc{Out} rules correspond to similar rules in the standard operational semantics of actors.

External actors may or may not participate in the DRL protocol themselves. It would be routine (but tedious) to define the composition of two DRL systems and to give additional rules for exchanging $\InfoMsg$ and $\ReleaseMsg$ messages between them. For simplicity, we only define the bare minimum interaction between a system and its environment; all $\ReleaseMsg$ and $\InfoMsg$ messages sent to external actors are simply dropped by the \textsc{ReleaseOut} and \textsc{InfoOut} events. For a complete formalization of actor composition, see \cite{aghaFoundationActorComputation1997}. We do, however, explore how snapshot aggregators from different actor systems can cooperate to detect cycles of terminated actors across the two systems (\Cref{sec:coop}).

The \textsc{In} event allows an external actor to send an application-level message to a receptionist $A$ containing a set of refobs $R$, all owned by $A$. If the external actor participates in DRL, the message is annotated as usual with the token $x$ used to send the message. Otherwise, a special $\NullToken$ token can be used instead. All targets in $R$ that are not internal actors are added to the set of external actors.

The \textsc{Out} event delivers an application-level message to an external actor with a set of refobs $R$. All internal actors referenced in $R$ become receptionists because their addresses have been exposed to the outside world.

\subsection{Basic Properties}

We now prove some basic properties of our model, both to help understand its semantics and to assist with later proofs.

\begin{lem}\label{lem:release-is-final}
If $B$ has undelivered messages along $\Refob x A B$, then $x$ is an unreleased refob.
\end{lem}
\begin{proof}
    There are three types of messages: $\AppMsg(x,-), \InfoMsg(x,-,-,-),$ and $\ReleaseMsg(x,-)$. All three messages can only be sent when $x$ is active. Moreover, the \textsc{Release} rule ensures that they must all be delivered before $x$ can be released.
\end{proof}

\begin{lem}\label{lem:facts-remain-until-cancelled}
$\ $
\begin{itemize}
    \item Once $\CreatedUsing(\Refob y A C, \Refob z B C)$ is added to $A$'s knowledge set, it will not be removed until after $A$ has sent an $\InfoMsg$ message containing $z$ to $C$.

    \item Once $\Created(\Refob z B C)$ is added to $C$'s knowledge set, it will not be removed until after $C$ has received the (unique) $\ReleaseMsg$ message along $z$.

    \item Once $\Released(\Refob z B C)$ is added to $C$'s knowledge set, it will not be removed until after $C$ has received the (unique) $\InfoMsg$ message containing $z$.
\end{itemize}
\end{lem}
\begin{proof}
    Immediate from the transition rules.
\end{proof}

The following lemma formalizes the argument made in \Cref{sec:overview}. In our model, it is possible for the message counts of two actor snapshots to agree, and yet for there to be undelivered messages between the two actors. However, in the special case where no messages are sent during the interval between the two snapshots, we can indeed trust the message counts to accurately reflect the number of undelivered messages.

\begin{lem}\label{lem:msg-counts}
    Consider a refob $\Refob x A B$.  Let $t_1, t_2$ be times such that $x$ has not yet been deactivated at $t_1$ and $x$ has not yet been released at $t_2$. In particular, $t_1$ and $t_2$ may be before the creation time of $x$.
    
    Suppose that $\alpha_{t_1}(A) \vdash \SentCount(x,n)$ and $\alpha_{t_2}(B) \vdash \RecvCount(x,m)$ and, if $t_1 < t_2$, that $A$ does not send any messages along $x$ during the interval $[t_1,t_2]$ . Then the difference $\max(n - m,0)$ is the number of messages sent along $x$ before $t_1$ that were not received before $t_2$.
\end{lem}
\begin{proof}
    Since $x$ is not deactivated at time $t_1$ and unreleased at time $t_2$, the message counts were never reset by the \textsc{SendRelease} or \textsc{Compaction} rules. Hence $n$ is the number of messages $A$ sent along $x$ before $t_1$ and $m$ is the number of messages $B$ received along $x$ before $t_2$. Hence $\max(n - m, 0)$ is the number of messages sent before $t_1$ and \emph{not} received before $t_2$.
\end{proof}

\subsection{Garbage}\label{sec:garbage-defn}

We can now operationally characterize actor garbage in our model. An actor $A$ can \emph{potentially receive a message} in $\kappa$ if there is a sequence of events (possibly of length zero) leading from $\kappa$ to a configuration $\kappa'$ in which $A$ has an undelivered message. We say that an actor is \emph{terminated} if it is idle and cannot potentially receive a message.

An actor is \emph{blocked} if it satisfies three conditions: (1) it is idle, (2) it is not a receptionist, and (3) it has no undelivered messages; otherwise, it is \emph{unblocked}. We define \emph{potential reachability} as the reflexive transitive closure of the potential acquaintance relation. That is, $A_1$ can potentially reach $A_n$ if and only if there is a sequence of unreleased refobs $(\Refob {x_1} {A_1} {A_2}), \dots, (\Refob {x_n} {A_{n-1}} {A_n})$; recall that a refob $\Refob x A B$ is unreleased if its target $B$ has not yet received a $\ReleaseMsg$ message for $x$.

Notice that an actor can potentially receive a message if and only if it is potentially reachable from an unblocked actor. Hence an actor is terminated if and only if it is only potentially reachable by blocked actors. A special case of this is \emph{simple garbage}, in which an actor is blocked and has no potential inverse acquaintances besides itself.

We say that a set of actors $S$ is \emph{closed} at time $t$ (with respect to the potential inverse acquaintance relation) if, whenever $B \in S$ and there is an unreleased refob $\Refob x A B$ at time $t$, then also $A \in S$. The \emph{closure} of a set of actors $S'$ is the smallest closed superset of $S'$. Notice that the closure of a set of terminated actors is also a set of terminated actors.

\section{Chain Lemma}\label{sec:chain-lemma}

To determine if an actor has terminated, one must show that all of its potential inverse acquaintances have terminated. This appears to pose a problem for termination detection, since actors cannot have a complete listing of all their potential inverse acquaintances without some synchronization: actors would need to consult their acquaintances before creating new references to them. In this section, we show that the DRL protocol provides a weaker guarantee that will nevertheless prove sufficient: knowledge about an actor's refobs is \emph{distributed} across the system and there is always a ``path'' from the actor to any of its potential inverse acquaintances.

\begin{figure}
    \centering
    \begin{tikzpicture}
	\begin{pgfonlayer}{nodelayer}
		\node [style=actor] (0) at (-6.5, 2.75) {$A_2$};
		\node [style=actor] (1) at (-11.5, 1) {$A_1$};
		\node [style=none] (17) at (-6.5, -3.75) {$\mathtt{Created}(x_1)$};
		\node [style=none] (51) at (1.5, 1) {\textbf{(2)}};
		\node [style=none] (54) at (-13.5, 1) {\textbf{(1)}};
		\node [style=actor] (57) at (-1.5, 1) {$A_3$};
		\node [style=none] (58) at (-3.25, 2.45) {$\mathtt{app}(y,\{x_3\})$};
		\node [style=actor] (82) at (-6.5, -2.5) {$B$};
		\node [style=none] (83) at (-11.5, 2.25) {$\mathtt{CreatedUsing}(x_1,x_2)$};
		\node [style=none] (84) at (11, 0) {$\mathtt{info}(x_2,x_3)$};
		\node [style=actor] (85) at (8.5, 2.75) {$A_2$};
		\node [style=actor] (86) at (3.5, 1) {$A_1$};
		\node [style=none] (87) at (8.5, -3.75) {$\mathtt{Created}(x_1)$};
		\node [style=actor] (88) at (13.5, 1) {$A_3$};
		\node [style=none] (89) at (11.75, 2.45) {$\mathtt{app}(y,\{x_3\})$};
		\node [style=actor] (90) at (8.5, -2.5) {$B$};
		\node [style=none] (91) at (3.5, 2.25) {$\mathtt{CreatedUsing}(x_1,x_2)$};
		\node [style=none] (93) at (-6.5, 4) {$\mathtt{CreatedUsing}(x_2,x_3)$};
		\node [style=none] (94) at (-9.5, -1) {$x_1$};
		\node [style=none] (97) at (-7, 0) {$x_2$};
		\node [style=none] (98) at (8, 0) {$x_2$};
		\node [style=none] (99) at (5.5, -1) {$x_1$};
	\end{pgfonlayer}
	\begin{pgfonlayer}{edgelayer}
		\draw [style=reference] (1) to (82);
		\draw [style=message] (0) to (57);
		\draw [style=reference] (86) to (90);
		\draw [style=message] (85) to (88);
		\draw [style=message, bend left=15] (85) to (90);
		\draw [style=reference] (0) to (82);
		\draw [style=reference] (85) to (90);
	\end{pgfonlayer}
\end{tikzpicture}
    \caption{An example of a chain from $B$ to $x_3$.}
    \label{fig:chain-example}
\end{figure}

Let us construct a concrete example of such a path, depicted by \Cref{fig:chain-example}. Suppose that $A_1$ spawns $B$, gaining a refob $\Refob{x_1}{A_1}{B}$. Then $A_1$ may use $x_1$ to create $\Refob{x_2}{A_2}{B}$, which $A_2$ may receive and then use $x_2$ to create $\Refob{x_3}{A_3}{B}$. 

At this point, there are unreleased refobs owned by $A_2$ and $A_3$ that are not included in $B$'s knowledge set. However, \Cref{fig:chain-example} shows that the distributed knowledge of $B,A_1,A_2$ creates a ``path'' to all of $B$'s potential inverse acquaintances. Since $A_1$ spawned $B$, $B$ knows the fact $\Created(x_1)$. Then when $A_1$ created $x_2$, it added the fact $\CreatedUsing(x_1, x_2)$ to its knowledge set, and likewise $A_2$ added the fact $\CreatedUsing(x_2, x_3)$; each fact points to another actor that owns an unreleased refob to $B$ (\Cref{fig:chain-example} (1)). 

Since actors can remove $\CreatedUsing$ facts by sending $\InfoMsg$ messages, we also consider (\Cref{fig:chain-example} (2)) to be a ``path'' from $B$ to $A_3$. But notice that, once $B$ receives the $\InfoMsg$ message, the fact $\Created(x_3)$ will be added to its knowledge set and so there will be a ``direct path'' from $B$ to $A_3$. We formalize this intuition with the notion of a \emph{chain} in a given configuration $\Config{\alpha}{\mu}{\rho}{\chi}$:
\begin{defi}\label{defn:chain}
A \emph{chain to $\Refob x A B$} is a sequence of unreleased refobs $(\Refob{x_1}{A_1}{B}),\allowbreak \dots,\allowbreak (\Refob{x_n}{A_n}{B})$ such that:
\begin{itemize}
    \item $\alpha(B) \vdash \Created(\Refob{x_1}{A_1}{B})$;
    \item For all $i < n$, either $\alpha(A_i) \vdash \CreatedUsing(x_i,x_{i+1})$ or the message $\Msg{B}{\InfoMsg(x_i,x_{i+1})}$ is in transit; and
    \item $A_n = A$ and $x_n = x$.
\end{itemize}
\end{defi}

\noindent We say that an actor $B$ is \emph{in the root set} if it is a receptionist or if there is an application message $\AppMsg(x,R)$ in transit to an external actor with $B \in \text{targets}(R)$. 

\begin{lem}[Chain Lemma]
\label{lem:chain-lemma}
Let $B$ be an internal actor in $\kappa$. If $B$ is not in the root set, then there is a chain to every unreleased refob $\Refob x A B$. Otherwise, there is a chain to some refob $\Refob y C B$ where $C$ is an external actor.
\end{lem}

\noindent\emph{Remark:}
When $B$ is in the root set, not all of its unreleased refobs are guaranteed to have chains. This is because an external actor may send $B$'s address to other receptionists without sending an $\InfoMsg$ message to $B$.

\begin{proof}
    We prove that the invariant holds in the initial configuration and at all subsequent times by induction on events $\kappa \Step e \kappa'$, omitting events that do not affect chains. Let $\kappa = \Config{\alpha}{\mu}{\rho}{\chi}$ and $\kappa' = \Config{\alpha'}{\mu'}{\rho'}{\chi'}$.
    
    In the initial configuration, the only refob to an internal actor is $\Refob y A A$. Since $A$ knows $\Created(\Refob{y}{A}{A})$, the invariant is satisfied.
    
    In the cases below, let $x,y,z,A,B,C$ be free variables, not referencing the variables used in the statement of the lemma.
        
    \begin{itemize}
        \item $\textsc{Spawn}(x,A,B)$ creates a new unreleased refob $\Refob x A B$, which satisfies the invariant because $\alpha'(B) \vdash \Created(\Refob x A B)$.

        \item $\textsc{Send}(x,\vec y, \vec z, A,B,\vec C)$ creates a set of refobs $R$. Let $(\Refob z B C) \in R$, created using $\Refob y A C$.
        
        If $C$ is already in the root set, then the invariant is trivially preserved. Otherwise, there must be a chain $(\Refob{x_1}{A_1}{C}), \dots, (\Refob{x_n}{A_n}{C})$ where $x_n = y$ and $A_n = A$. Then $x_1,\dots,x_n,z$ is a chain in $\kappa'$, since $\alpha'(A_n) \vdash \CreatedUsing(x_n,z)$. 
        
        If $B$ is an internal actor, then this shows that every unreleased refob to $C$ has a chain in $\kappa'$. Otherwise, $C$ is in the root set in $\kappa'$. To see that the invariant still holds, notice that $\Refob z B C$ is a witness of the desired chain.
        
        \item $\textsc{SendInfo}(y,z,A,B,C)$ removes the $\CreatedUsing(y,z)$ fact but also sends $\InfoMsg(y,z,B)$, so chains are unaffected.
        
        \item $\textsc{Info}(y,z,B,C)$ delivers $\InfoMsg(y,z,B)$ to $C$ and adds $\Created(\Refob z B C)$ to its knowledge set.
        
        Suppose $\Refob z B C$ is part of a chain $(\Refob{x_1}{A_1}{C}), \dots, (\Refob{x_n}{A_n}{C})$, i.e. $x_i = y$ and $x_{i+1} = z$ and $A_{i+1} = B$ for some $i < n$. Since $\alpha'(C) \vdash \Created(\Refob{x_{i+1}}{A_{i+1}}{C})$, we still have a chain $x_{i+1},\dots,x_n$ in $\kappa'$.
        
        \item $\textsc{Release}(x,A,B)$ releases the refob $\Refob x A B$. Since external actors never release their refobs, both $A$ and $B$ must be internal actors.
        
        Suppose the released refob was part of a chain $(\Refob{x_1}{A_1}{B}), \dots, (\Refob{x_n}{A_n}{B})$, i.e. $x_i = x$ and $A_i = A$ for some $i < n$. We will show that $x_{i+1},\dots,x_n$ is a chain in $\kappa'$.
        
        Before performing $\textsc{SendRelease}(x_i,A_i,B)$, $A_i$ must have performed the $\textsc{Info}(x_i,x_{i+1},\allowbreak A_{i+1},B)$ event. Since the $\InfoMsg$ message was sent along $x_i$, Lemma~\ref{lem:release-is-final} ensures that the message must have been delivered before the present \textsc{Release} event. Furthermore, since $x_{i+1}$ is an unreleased refob in $\kappa'$, Lemma~\ref{lem:facts-remain-until-cancelled} ensures that $\alpha'(B) \vdash \Created(\Refob{x_{i+1}}{A_{i+1}}{B})$.
        
        \item $\textsc{In}(A,R)$ adds a message from an external actor to the internal actor $A$. This event can only create new refobs that point to receptionists, so it preserves the invariant.

        \item $\textsc{Out}(x,B,R)$ emits a message $\AppMsg(x,R)$ to the external actor $B$. Since all targets in $R$ are already in the root set, the invariant is preserved. \qedhere
    \end{itemize}
\end{proof}

\noindent An immediate application of the Chain Lemma is to allow actors to detect when they are simple garbage. If any actor besides $B$ owns an unreleased refob to $B$, then $B$ must have a fact $\Created(\Refob x A B)$ in its knowledge set where $A \ne B$. Hence, if $B$ has no such facts, then it must have no nontrivial potential inverse acquaintances. Moreover, since actors can only have undelivered messages along unreleased refobs, $B$ also has no undelivered messages from any other actor; it can only have undelivered messages that it sent to itself. This gives us the following result:

\begin{thm}
    Suppose $B$ is idle with knowledge set $\Phi$, such that:
    \begin{itemize}
        \item $\Phi$ does not contain any facts of the form $\Created(\Refob x A B)$ where $A \ne B$; and
        \item for all facts $\Created(\Refob x B B) \in \Phi$, also $\Phi \vdash \SentCount(x,n) \land \RecvCount(x,n)$ for some $n$.
    \end{itemize}
    Then $B$ is simple garbage.
\end{thm}

\section{Termination Detection}\label{sec:termination-detection}

In the following three sections, we present the scheme for detecting non-simple terminated actors in DRL. First, we define what it means for a set of snapshots to be \emph{finalized} and prove that finalized sets correspond to closed sets of terminated actors. This reduces termination detection to simply collecting snapshots at an aggregator and periodically searching the collection for finalized subsets. We prove that such an approach is safe and live (under reasonable fairness assumptions). Next, in \Cref{sec:maximal}, we give an algorithm for finding the \emph{maximum} finalized subset---the union of all finalized subsets---in an arbitrary set of snapshots. This gives each aggregator an efficient procedure for detecting terminated actors. Lastly, in \Cref{sec:coop}, we show how a decentralized group of snapshot aggregators can cooperate to detect distributed garbage while exchanging minimal information. This makes DRL's termination detection scalable, parallelizable, and capable of making progress despite network partitions.

\subsection{Consistent snapshots}

Recall that when we speak of a set of snapshots $Q$, we assume each snapshot was taken by a different actor. We will therefore represent $Q$ as a mapping from actor names to snapshots, with $Q(A)$ denoting $A$'s snapshot in $Q$.

As shown in \Cref{sec:overview}, actor snapshots taken at different times can result in conflicting accounts of the configuration. Hence, in general, an arbitrary set of snapshots $Q$ does not accurately describe the current configuration. If a set of snapshots \textit{does} accurately describe the configuration, we say that it is \textit{consistent}. Formally:
\begin{defi}
    $Q$ is consistent at time $t$ when $\forall \phi, \forall A \in \dom(Q), \ Q(A) \vdash \phi$ if and only if $\alpha_{t}(A) \vdash \phi$.
\end{defi}

That is, the snapshot $Q(A)$ may not have been taken at time $t$---yet the contents of $A$'s knowledge set at time $t$ are the same as $Q(A)$. If $Q$ is consistent at time $t$, then it is as if all the actors of $\dom(Q)$ took their snapshots at time $t$.

Another important notion is that of a terminated actor's \emph{final action}. We define this as the last non-snapshot event that an actor performs before becoming terminated. Notice that an actor's final action can only be an \textsc{Idle}, \textsc{Info}, or \textsc{Release} event. This is because terminated actors are idle, and only these three events change an actor's status from busy to idle.  Note also that the final action may come \emph{strictly before} an actor becomes terminated, since a blocked actor is only considered to be terminated once all of its potential inverse acquaintances are terminated.

We can now give a simple proof of our earlier claim that snapshots from terminated actors are consistent. This property will allow us to treat finalized sets of snapshots as if they were all taken at an instant in global time.

\begin{lem}\label{lem:terminated-is-consistent}
    Let $S$ be a closed set of terminated actors at time $t_f$. If every actor in $S$ took a snapshot sometime after its final action, then the resulting set of snapshots $Q$ is consistent at $t_f$.
\end{lem}
\begin{proof}
    A terminated actor's knowledge set never changes. Moreover, an actor's knowledge set cannot change between the point of performing its final action and becoming terminated. Hence each $A$'s snapshot in $Q$ agrees with its knowledge set at time $t_f$.
\end{proof}

\subsection{Finalized sets}

Given a consistent set of snapshots $Q$, is it possible to determine whether the actors $\dom(Q)$ are terminated? Let us begin by giving an alternative characterization of terminated actors.

\begin{lem}
Let $A$ be an actor at time $t$ and let $S$ be the \emph{closure} of $\{A\}$. Then $A$ is terminated if and only if the actors of $S$ are idle and \emph{mutually quiescent}, i.e.~there are no undelivered messages between actors of $S$.
\end{lem}
\begin{proof}
By definition, $A$ is terminated if and only if every actor that can potentially reach $A$ is blocked. Notice that the closure $S$ is precisely the set of actors that can potentially reach $A$. Hence, if $A$ is terminated then the actors of $S$ are blocked, i.e.~idle and have no undelivered messages from any actor.

Conversely, let the actors of $S$ be idle and mutually quiescent. Could any of them have undelivered messages from actors outside $S$? The basic property Lemma~\ref{lem:release-is-final} shows that this cannot occur; any sender of such a message must be in $S$. Hence the actors of $S$ are all blocked and $A$ is terminated.
\end{proof}

By the above lemma, we can only conclude that the actors of $\dom(Q)$ are terminated if $\dom(Q)$ is closed, mutually quiescent, and all the actors are idle. This last condition is automatically satisfied by our communication protocol, since only idle actors take snapshots. In order to check the other two conditions, we will inspect the snapshots themselves.

To check that $\dom(Q)$ is closed, recall that every unreleased refob $\Refob x A B$ has a chain, $(\Refob{x_1}{A_1}{B}),\dots,(\Refob{x_n}{A_n}{B})$. If $\dom(Q)$ is terminated, then there can be no undelivered $\InfoMsg$ messages. Hence $x$ must satisfy the following predicate:

\begin{defi}
Let $Q \vdash \Chain(\Refob x A B)$ if there exist $(\Refob{x_1}{A_1}{B}),\dots,(\Refob{x_n}{A_n}{B})$ such that:
\begin{enumerate}
    \item $Q \vdash \Created(x_1)$ and $Q \not\vdash \Released(x_1)$;
    \item For all $i < n$, $Q \vdash \CreatedUsing(x_i,x_{i+1})$ and $Q \not\vdash \Released(x_{i+1})$;
    \item $A_n = A$ and $x_n = x$.
\end{enumerate}
\end{defi}

\noindent Checking where $\dom(Q)$ is closed therefore amounts to checking that $B \in Q$ and $Q \vdash \Chain(\Refob x A B)$ always implies $A \in Q$.

To decide whether $\dom(Q)$ is mutually quiescent, we need to check that there are no undelivered messages along each unreleased refob $\Refob x A B$ where $A,B \in \dom(Q)$. In a consistent set of snapshots, we can do so by inspecting the message counts of $A$ and $B$ for $x$. Hence $x$ must satisfy the following predicate:

\begin{defi}
Let $Q \vdash \Relevant(\Refob x A B)$ if there exists $n$ such that $Q \vdash \Activated(x) \land \SentCount(x,n) \land \RecvCount(x,n)$.
\end{defi}

Together, we can use these predicates to characterize a set of snapshots from a closed, terminated set of actors.
\begin{defi}\label{defn:weakly-finalized}
A set of snapshots $Q$ is \emph{finalized} if, for all $B \in \dom(Q)$ and for all $\Refob x A B$, 
\begin{enumerate}
    \item $Q \vdash \Chain(x)$ implies $A \in Q$; and
    \item $Q \vdash \Chain(x)$ implies $Q \vdash \Relevant(x)$.
\end{enumerate}
\end{defi}
\noindent The first condition ensures that $\dom(Q)$ is closed; $Q \vdash \Chain(\Refob x A B)$ implies that $A$ is a potential inverse acquaintance of $B$. The second condition ensures that $\dom(Q)$ is mutually quiescent: between any two actors $A,B \in \dom(Q)$, the message counts for any unreleased refob $\Refob x A B$ must agree.

If finalized sets $Q$ correspond to closed sets of terminated actors $S$, then we would expect that a consistent snapshot of $S$ is finalized. This is indeed the case:

\begin{thm}\label{lem:terminated-is-complete}
    Let $Q$ be a consistent set of snapshots at time $t$ of a closed set of terminated actors $S$. Then $Q$ is finalized.
\end{thm}
\begin{proof}
    First, we show that if $B \in \dom(Q)$ and $\Refob x A B$ is an unreleased refob at time $t$, then $Q \vdash \Chain(x) \land \Activated(x) \land \SentCount(x,n) \land \RecvCount(x,n)$ for some $n$.
    
    \begin{itemize}
        \item $Q \vdash \Chain(x)$ follows from Lemma~\ref{lem:chain-lemma} because $B$ is blocked and $S$ is closed.
    
        \item $Q \vdash \Activated(x)$ holds because $x$ must be activated: if $x$ were pending then $A$ would be unblocked and if $x$ were deactivated then $B$ would be unblocked.
    
        \item $Q \vdash \SentCount(x,n) \land \RecvCount(x,n)$ holds because there are no undelivered messages between $A$ and $B$ at time $t$, so the send and receive counts of $x$ at time $t$ must agree.
    \end{itemize}
    Now it suffices to show that, if $Q \vdash \Chain(x)$, then $\Refob x A B$ is unreleased at time $t$. There are two cases: Either $Q(B) \vdash \Created(x)$ or $Q(C) \vdash \CreatedUsing(y,x)$ for some $C,y$. In both cases, $x$ has been created before time $t$. Since $Q$ is consistent and $Q(B) \not\vdash \Released(x)$, it follows from Lemma~\ref{lem:facts-remain-until-cancelled} that $B$ has not yet received a $\ReleaseMsg$ message for $x$. Hence $x$ is unreleased at time $t$.
\end{proof}

By contrapositive, if $Q$ is \emph{not} finalized then it cannot be a consistent set of snapshots from a closed set of terminated actors. Recall also that any set of snapshots from terminated actors is guaranteed to be consistent (Lemma~\ref{lem:terminated-is-consistent}). Hence, if $Q$ is not finalized, then either some actor in $Q$ is not terminated or $\dom(Q)$ is not closed. The latter case indicates that there is insufficient information to conclude whether the actors of $\dom(Q)$ are terminated; they may or may not be reachable by an unblocked actor outside of $\dom(Q)$.

We now show that, surprisingly, the converse of Theorem~\ref{lem:terminated-is-complete} also holds:
any finalized set of snapshots $Q$ necessarily describes a closed set of terminated actors, with each snapshot taken some point after the actor's final action. By Lemma~\ref{lem:terminated-is-consistent}, such a set of snapshots is also consistent. 

Given a set of snapshots $Q$ taken before some time $t_f$, we write $Q_t$ to denote those snapshots in $Q$ that were taken before time $t < t_f$. If $A \in \dom(Q)$, we denote the time of $A$'s snapshot as $t_A$.

\begin{thm}\label{thm:safety-helper}
Let $Q$ be a finalized set of snapshots at time $t_f$. Then for all times $t$:
\begin{enumerate}
    \item If $B \in \dom(Q_t)$ and $\Refob x A B$ is unreleased, then $Q \vdash \Chain(x)$.
    \item The actors of $Q_t$ are all blocked.
\end{enumerate}
In particular $Q_t = Q$, when $t \ge t_f$.
\end{thm}
\begin{proof}
Proof by induction on $t$. Notice that these two properties trivially hold in the initial configuration because $Q_0 = \emptyset$. 

For the induction step, assume both properties hold at time $t$ and call them IH-1 and IH-2, respectively. We show that IH-1 and IH-2 are preserved by any legal transition $\kappa \Step e \kappa'$.

\paragraph*{$\textsc{Snapshot}(B, \Phi)$}

Suppose \(B \in \dom(Q)\) takes a snapshot at time \(t\). We show that if $\Refob x A B$ is unreleased at time $t$, then $Q \vdash \Chain(x)$ and there are no undelivered messages along $x$ from $A$ to $B$. We do this with the help of two lemmas.

\begin{lem}\label{lem:complete-ref}
    If $Q \vdash \Chain(\Refob x A B)$, then $x$ is unreleased at time $t$ and there are no undelivered messages along $x$ at time $t$. Moreover, if $t_A > t$, then there are no undelivered messages along $x$ throughout the interval $[t,t_A]$.
\end{lem}
\begin{proof}[Proof (Lemma)]
    Since $Q \vdash \Relevant(\Refob x A B)$, we have $A \in \dom(Q)$ and $Q \vdash \Activated(x)$ and $Q \vdash \SentCount(x,n) \land \RecvCount(x,n)$ for some $n$.
    
    Consider the case when $t_A > t$. Since $Q(A) \vdash \Activated(x)$, $x$ is not deactivated and therefore not released at $t_A$ or $t$. Hence, by Lemma~\ref{lem:msg-counts}, every message sent along $x$ before $t_A$ was received before $t$. Since message sends precede receipts, each of those messages was sent before $t$. Hence there are no undelivered messages along $x$ throughout $[t,t_A]$.
    
    Now consider the case when $t_A < t$. Since $Q(A) \vdash \Activated(x)$, $x$ is not deactivated and not released at $t_A$. By IH-2, $A$ was blocked throughout the interval $[t_A,t]$, so it could not have sent a $\ReleaseMsg$ message. Hence $x$ is still not deactivated at $t$ and therefore not released at $t$. By Lemma~\ref{lem:msg-counts}, all messages sent along $x$ before $t_A$ must have been delivered before $t$. Hence, there are no undelivered messages along $x$ at time $t$.
\end{proof}

\begin{lem}\label{lem:complete-chains}
    Let $\Refob{x_1}{A_1}{B}, \dots, \Refob{x_n}{A_n}{B}$ be a chain to $\Refob x A B$ at time $t$. Then $Q \vdash \Chain(x)$.
\end{lem}
\begin{proof}[Proof (Lemma)]
    We prove by induction on the length of the chain that $Q \vdash \Chain(x_i)$ for all $i \le n$.
    
    \noindent\textbf{Base case:} By the definition of a chain, $\alpha_t(B) \vdash \Created(x_1)$ and $\alpha_t(B) \not\vdash \Released(x_1)$. Since $B$'s snapshot happens at time $t$, we must have $Q(B) \vdash \Created(x_1)$ and $Q(B) \not\vdash \Released(x_1)$.
    
    \noindent\textbf{Induction step:} Assume $Q \vdash \Chain(x_i)$. Notice that $Q \not\vdash \Released(x_{i+1})$ because $x_{i+1}$ is unreleased at the time of $B$'s snapshot. Hence, it suffices to show that $Q \vdash \CreatedUsing(x_i,x_{i+1})$.
    
    Since $Q \vdash \Relevant(x_i)$, we must have $A_i \in \dom(Q)$. Let $t_i$ be the time of $A_i$'s snapshot; we will show $\alpha_{t_i}(A_i) \vdash \CreatedUsing(x_i,x_{i+1})$.
    
    By the definition of a chain, either the message $\Msg{B}{\InfoMsg(x_i,x_{i+1})}$ is in transit at time $t$, or $\alpha_t(A_i) \vdash \CreatedUsing(x_i,x_{i+1})$. But the first case is impossible by Lemma~\ref{lem:complete-ref}, so we only need to consider the latter.
    
    Consider the case where $t_i > t$. Lemma~\ref{lem:complete-ref} implies that $A_i$ cannot perform the $\textsc{SendInfo}(x_i,x_{i+1},A_{i+1},B)$ event during $[t,t_i]$. Hence $\alpha_{t_i}(A_i) \vdash \CreatedUsing(x_i,x_{i+1})$.
    
    Now consider the case where $t_i < t$. By IH-2, $A_i$ must have been blocked throughout the interval $[t_i,t]$. Hence $A_i$ could not have created any refobs during this interval, so $x_{i+1}$ must have been created before $t_i$. This implies $\alpha_{t_i}(A_i) \vdash \CreatedUsing(x_i,x_{i+1})$.
\end{proof}

Lemma~\ref{lem:complete-chains} implies that $B$ cannot be in the root set. If it were, then by the Chain Lemma there would be a refob $\Refob y C B$ with a chain where $C$ is an external actor. Since $Q \vdash \Chain(y)$, there would need to be a snapshot from $C$ in $Q$---but external actors do not take snapshots, so this is impossible.

Since $B$ is not in the root set, there must be a chain to every unreleased refob $\Refob x A B$. By Lemma~\ref{lem:complete-chains}, $Q \vdash \Chain(x)$. By Lemma~\ref{lem:complete-ref}, there are no undelivered messages to $B$ along $x$ at time $t$. Since $B$ can only have undelivered messages along unreleased refobs (Lemma~\ref{lem:release-is-final}), the actor is indeed blocked.

\paragraph*{$\textsc{Send}(x,\vec y, \vec z, A,B,\vec C)$}

In order to maintain IH-2, we must show that if $B \in \dom(Q_t)$ then this event cannot occur. So suppose $B \in \dom(Q_t)$. By IH-1, we must have $Q \vdash \Chain(\Refob x A B)$, so $A \in \dom(Q)$. By IH-2, we moreover have $A \not\in \dom(Q_t)$---otherwise $A$ would be blocked and unable to send this message. Since $Q \vdash \Relevant(x)$, we must have $Q(A) \vdash \SentCount(x,n)$ and $Q(B) \vdash \RecvCount(x,n)$ for some $n$. Hence $x$ is not deactivated at $t_A$ and unreleased at $t_B$. By Lemma~\ref{lem:msg-counts}, every message sent before $t_A$ is received before $t_B$. Hence $A$ cannot send this message to $B$ because $t_A > t > t_B$.

In order to maintain IH-1, we will show that if one of the refobs sent to $B$ in this step is $\Refob z B C$, where $C \in \dom(Q_t)$, then $Q \vdash \Chain(z)$. In the configuration that follows this $\textsc{Send}$ event, $\CreatedUsing(y,z)$ occurs in $A$'s knowledge set. By the same argument as above, $A \in \dom(Q) \setminus Q_t$ and $Q(A) \vdash \SentCount(y,n)$ and $Q(C) \vdash \RecvCount(y,n)$ for some $n$. Hence $A$ cannot perform the $\textsc{SendInfo}(y,z,A,B,C)$ event before $t_A$, so $Q(A) \vdash \CreatedUsing(y,z)$. Since $Q \vdash \Chain(y) \land \CreatedUsing(y,z)$ and $Q \not\vdash \Released(z)$, we have $Q \vdash \Chain(z)$.

\paragraph*{$\textsc{SendInfo}(y,z,A,B,C)$}

By the same argument as above, $A \not\in \dom(Q_t)$ cannot send an $\InfoMsg$ message to $B \in \dom(Q_t)$ without violating message counts, so IH-2 is preserved.

\paragraph*{$\textsc{SendRelease}(x,A,B)$}

Suppose that $A \not\in \dom(Q_t)$ and $B \in \dom(Q_t)$. By IH-1, $Q \vdash \Chain(x)$ at time $t$. Since $Q \vdash \Relevant(x)$, it follows that $Q(A) \vdash \Activated(x)$. Hence $A$ cannot deactivate $x$ and IH-2 is preserved.

\paragraph*{$\textsc{In}(A,R)$}

By IH-1, every potential inverse acquaintance of an actor in $Q_t$ is also in $Q$. Hence none of the actors in $Q_t$ is a receptionist and this rule does not affect the invariants.

\paragraph*{$\textsc{Out}(x,B,R)$}
Suppose $(\Refob y B C) \in R$ where $C \in \dom(Q_t)$. Then $y$ is unreleased and $Q \vdash \Chain(y)$ and $B \in \dom(Q)$. But this is impossible because $B$ is an external actor and external actors do not take snapshots.
\let\qed\relax
\hfill {\small \em End of proof of \cref{thm:safety-helper}.}~\qedsymbol
\end{proof}

\begin{cor}[Safety]
\label{cor:safety}
If $Q$ is a finalized set of snapshots at time $t_f$ then the actors in $Q$ are all terminated at $t_f$.
\end{cor}
\begin{proof}
    \Cref{thm:safety-helper} implies that, at $t_f$, all the actors in $Q$ are blocked. Together with the fact that $Q$ is finalized, it also implies that $Q$ is closed under the potential inverse acquaintance relation at $t_f$. Hence every actor that can potentially reach $B \in \dom(Q)$ at $t_f$ is blocked, so by definition every $B \in \dom(Q)$ is terminated at $t_f$.
\end{proof}

Recall that a snapshot aggregator detects terminated actors by receiving actor snapshots and periodically looking for finalized subsets. It is now simple to see that this algorithm is live, under reasonable fairness assumptions:

\begin{thm}[Liveness]\label{thm:liveness}
If every actor eventually takes a snapshot after performing an \textsc{Idle}, \textsc{Info}, or \textsc{Release} event, then every terminated actor is eventually part of a finalized set of snapshots.
\end{thm}
\begin{proof}
    If an actor $A$ is terminated, then the closure $S$ of $\{A\}$ is a terminated set of actors. Every actor eventually takes a snapshot after taking its final action and the resulting set of snapshots is consistent, by Lemma~\ref{lem:terminated-is-consistent}. Then Theorem~\ref{lem:terminated-is-complete} implies that the resulting snapshots is finalized.
\end{proof}

\subsection{Strongly finalized sets}

Note that our definition of finalized sets differs from the definition which originally appeared in \cite{DBLP:conf/concur/PlyukhinA20}. This old definition, which we now call ``strongly finalized'', used $Q \vdash \Unreleased(x)$ instead of $Q \vdash \Chain(x)$:
\begin{defi}\label{defn:finalized}
A set of snapshots $Q$ is \emph{strongly finalized} if, for all $B \in \dom(Q)$ and for all $\Refob x A B$, $Q \vdash \Unreleased(x)$ implies $Q \vdash \Relevant(x)$.
\end{defi}
In fact, the two notions of finalized are equivalent.  However, in the process of developing the theory in \Cref{sec:maximal,sec:coop}, we found this old definition to be inconvenient.

Notice that any strongly finalized $Q$ is also finalized because $Q \vdash \Chain(x)$ implies $Q \vdash \Unreleased(x)$. However, in an arbitrary set $Q$, $Q \vdash \Unreleased(x)$ does not imply $Q \vdash \Chain(x)$; there could exist $A,B,C \in \dom(Q)$ such that $Q(A) \vdash \CreatedUsing(\Refob x A C,\allowbreak \Refob y B C)$ but $Q \not\vdash \Chain(x)$. Could such a situation occur in a finalized $Q$? We prove below that no, this is impossible:

\begin{thm}
If $Q$ is finalized then $Q$ is strongly finalized.
\end{thm}
\begin{proof}
    We will show that if $Q$ is finalized, then $Q \vdash \Unreleased(\Refob x A B)$ and $B \in \dom(Q)$ implies $Q \vdash \Chain(\Refob x A B)$. 
    
    Note that $Q$ is consistent at some time $t_f$, since $\dom(Q)$ is a closed terminated set of actors where each snapshot was taken after the actor's final action.
    
    By definition, $Q \vdash \Unreleased(\Refob x A B)$ implies that $Q \not\vdash \Released(x)$ and either (1) $Q \vdash \Created(x)$  or (2) $Q \vdash \Chain(y) \land \CreatedUsing(y,x)$ for some $\Refob y C B$.
    
    In case (1), we have a trivial chain $Q \vdash \Chain(x)$.
    
    In case (2), notice that we must have $Q \vdash \Activated(y)$ because $Q \vdash \Relevant(y)$. Since $Q$ is consistent, $y$ must be an unreleased refob at $t_f$. Due to Lemma~\ref{lem:chain-lemma}, there must be a chain of unreleased refobs $(\Refob{y_1}{C_1}{B}),\dots,(\Refob{y_n}{C_n}{B})$ at $t_f$. Again since $Q$ is consistent, we must have $Q \vdash \Created(y_1)$ and $Q \vdash \CreatedUsing(y_i,y_{i+1})$ for each $i < n$ and $Q \not\vdash \Released(y_i)$ for each $i \le n$. Hence $Q \vdash \Chain(y)$. Since also $Q \vdash \CreatedUsing(y,x)$ and $Q \not\vdash \Released(x)$, we can ``extend'' this chain to derive $Q \vdash \Chain(x)$.
\end{proof}

Thanks to this theorem, we can use the two definitions interchangeably.

\section{Maximal Finalized Subsets}\label{sec:maximal}

In the previous section, we showed that a finalized set of snapshots corresponds to a closed set of terminated actors. Hence, the problem of garbage collection reduces to finding all the finalized subsets of an arbitrary set of snapshots $Q$. In this section, we show that there is in fact a single largest finalized subset $Q_f \subseteq Q$ that contains all other finalized subsets. We then show that $Q_f$ can be computed in linear time by removing all snapshots that cannot appear in a finalized subset.

Originally, we presented a slightly different algorithm in \cite{DBLP:conf/concur/PlyukhinA20}. It operates similarly to the new algorithm, iteratively removing snapshots that appear not to be in a finalized subset. However, we subsequently discovered through model checking that the original algorithm could sometimes be overzealous: the presence of certain ``stale'' snapshots in $Q$ can cause other snapshots to be unnecessarily removed. In other words, the computed set is finalized but not necessarily maximal. Nevertheless, since the original algorithm can have better cache locality than the new algorithm, it may be more practical for real systems. We present the algorithm again in \Cref{sec:heuristic} and go on to prove that it \emph{eventually} detects all terminated actors, under reasonable fairness conditions.

\subsection{Chain Algorithm}\label{sec:chain-algorithm}

Let us first show that a maximum finalized subset always exists. Notice that finalized sets are closed under union when they agree on $\dom(Q_1) \cap \dom(Q_2)$:
\begin{lem}\label{lem:weakly-finalized-closed-union}
    Let $Q_1,Q_2$ be finalized sets of snapshots that agree at their intersection, i.e.~ $\forall A \in \dom(Q_1) \cap \dom(Q_2),\ Q_1(A) = Q_2(A)$. Then $Q_1 \cup Q_2$ is also finalized.
\end{lem}
\begin{proof}
    Suppose there exists $\Refob x A B$ such that $Q_1 \cup Q_2 \vdash \Chain(x)$ and $Q_1 \cup Q_2 \not\vdash \Relevant(x)$. Let $(\Refob{x_1}{A_1}{B}),\dots,(\Refob{x_n}{A_n}{B})$ be the chain.
    
    Let $Q$ be either $Q_1$ or $Q_2$; we prove by induction on $n$ that, if $B \in \dom(Q)$, then $\forall i \le n,\ Q \vdash \Chain(x_i)$. If $n = 1$ then $Q(B) \vdash \Created(x_1)$ and $Q(B) \not\vdash \Released(x_1)$ implies $Q \vdash \Chain(x_1)$. For $n > 1$, $Q \vdash \Chain(x_{n-1})$ implies $Q \vdash \Relevant(x_{n-1})$ since $Q$ is finalized, and therefore $A_{n-1} \in \dom(Q)$. Hence $Q(A_{n-1}) \vdash \CreatedUsing(x_{n-1},x_n)$, which implies $Q \vdash \Chain(x_n)$.
    
    Since each $Q_1,Q_2$ is finalized, we must therefore have $Q \vdash \Relevant(x)$, i.e.~$Q \vdash \Activated(x) \land \SentCount(x,n) \land \RecvCount(x,n)$ for some $n$. By definition of $(\vdash)$, it follows that $Q_1 \cup Q_2 \vdash \Relevant(x)$.
\end{proof}
Any two finalized subsets of $Q$ will satisfy the condition of this lemma. Hence the maximum finalized subset $Q_f$ is the union of all finalized subsets of $Q$.

Next, we characterize which snapshots in $Q$ \emph{can} and \emph{cannot} appear in a finalized subset of $Q$. To this end, we define the following useful concept:
\begin{defi}
    We say that $B$ \emph{depends on} $A$ in $Q$ if $A = B$ or there is a sequence of one or more refobs $(\Refob{x_1}{A_1}{A_2}),\dots,(\Refob{x_n}{A_{n-1}}{A_n})$ where $A = A_1$ and $B = A_n$ and, for each $i < n$, $Q \vdash \Chain(\Refob{x_i}{A_i}{A_{i+1}})$. Hence the ``depends on'' relation is reflexive and transitive.
\end{defi}
The following lemmas show that if there exists $A \in \dom(Q)$ such that $A \not\in \dom(Q_f)$, then every $B$ that depends on $A$ in $Q$ also cannot appear in $Q_f$.
\begin{lem}
    If $Q \vdash \Chain(\Refob x A B)$ then, for any finalized subset $Q_f$ of $Q$, if $B \in \dom(Q_f)$ then $Q_f \vdash \Chain(\Refob x A B)$.
\end{lem}
\begin{proof}
    Proof by induction on the length of the chain $(\Refob{x_1}{A_1}{B}),\dots,(\Refob{x_n}{A_n}{B})$.
    
    If $n = 1$ then $Q(B) \vdash \Created(x_1)$ and $Q(B) \not\vdash \Released(x_1)$. Hence any subset $Q_f$ of $Q$ must have $Q_f \vdash \Chain(x_1)$.
    
    If $n > 1$, assume $Q_f \vdash \Chain(x_{n-1})$. Since $Q_f$ is finalized, $A_{n-1} \in \dom(Q_f)$. Since $Q(A_{n-1}) \vdash \CreatedUsing(x_{n-1},x_n)$ and $Q(B) \not\vdash \Released(x_n)$, it follows that $Q_f \vdash \Chain(x_n)$.
\end{proof}

\begin{lem}
    If $B$ depends on $A$ in $Q$, then every finalized subset of $Q$ containing $B$ must also contain $A$.
\end{lem}
\begin{proof}
    If $A = B$ then the lemma trivially holds. We prove that this must hold for nontrivial sequences by induction on the length of the sequence $(\Refob{x_1}{A_1}{A_2})$, $\dots$, $(\Refob{x_n}{A_{n-1}}{A_n})$.
    
    If $n = 1$ then $Q \vdash \Chain(\Refob x A B)$. Then $Q_f \vdash \Chain(\Refob x A B)$ for any finalized subset $Q_f$ containing $B$. Since $Q_f$ is finalized, we must also have $Q_f \vdash \Relevant(\Refob x A B)$ and therefore $A \in \dom(Q_f)$.
    
    For $n > 1$, assume any finalized subset containing $A_{n-1}$ must also contain $A_1$. By the same argument as above, any finalized subset containing $A_n$ must contain $A_{n-1}$ and therefore also contain $A_1$.
\end{proof}

We can also use the notion of dependency to give a new characterization of finalized sets:
\begin{defi}
    $C$ is \emph{finalized} in $Q$ if, for all $B$ on which $C$ depends, for all $\Refob x A B$, $Q \vdash \Chain(x)$ implies $Q \vdash \Relevant(x)$.
\end{defi}
\begin{lem}
    $C$ is finalized in $Q$ if and only if $C$ is in a finalized subset of $Q$.
\end{lem}
\begin{proof}
    If $C$ is finalized in $Q$, let $Q_f$ be a subset of $Q$ containing only snapshots of actors on which $C$ depends. To see that $Q_f$ is finalized, first notice that each $B \in \dom(Q_f)$ has a sequence of refobs $(\Refob{x_1}{A_1}{A_2})$, $\dots$, $(\Refob{x_n}{A_{n-1}}{A_n})$ where $B = A_1$ and $C = A_n$ and $Q \vdash \Chain(\Refob{x_i}{A_i}{A_{i+1}}$ for each $i < n$. For any $(\Refob x A B)$, $Q_f \vdash \Chain(\Refob x A B)$ implies $Q \vdash \Chain(\Refob x A B)$ and therefore $C$ depends on $A$ in $Q$. Hence $Q \vdash \Relevant(x)$ and therefore $Q_f \vdash \Relevant(x)$.
    
    Conversely, let $C$ be in a finalized subset $Q_f$ and consider a sequence $(\Refob{x_1}{A_1}{A_2})$, $\dots$, $(\Refob{x_n}{A_{n-1}}{A_n})$ where $Q \vdash \Chain(\Refob{x_i}{A_i}{A_{i+1}}$ for each $i < n$. Then $A_i \in \dom(Q_f)$ for each $i \le n$. Since $Q_f$ is finalized, $Q_f \vdash \Relevant(x_i)$ for each $i < n$. Hence $Q \vdash \Relevant(x_i)$ for each $i < n$.
\end{proof}

Since $C \in Q_f$ if and only if $C$ is finalized in $Q$, it follows that $C \not\in Q_f$ if and only if $C$ is not finalized in $Q$. Hence, to find the maximum finalized subset of $Q$ it suffices to remove every snapshot that is not finalized in $Q$.

\begin{algorithm}
\caption{Compute the largest finalized subset of $Q$}\label{alg:maximum-finalized-subset}
\begin{algorithmic}[1]
    \State Let $S_1 \subseteq \dom(Q)$ be the set of all actors $B$ for which there exists $\Refob x A B$ such that $Q \vdash \Chain(x)$ and $Q \not\vdash \Relevant(x)$.
    \State Let $S_2 \subseteq \dom(Q)$ be the set of all actors that depend on actors in $S_1$.
    \State Let $S_3 = \dom(Q) \setminus S_2$.
\end{algorithmic}
\end{algorithm}

\begin{thm}
    \Cref{alg:maximum-finalized-subset} computes the largest finalized subset of $Q$. 
\end{thm}
\begin{proof}
    Clearly, an actor is not finalized in $Q$ if it depends on one of the actors of $S_1$. Hence $S_2$ is precisely the set of all actors that are not finalized in $Q$. Its complement, $S_3$, is therefore the set of all finalized actors in $Q$.
\end{proof}

\subsection{Heuristic algorithm}\label{sec:heuristic}

Although the algorithm above has $O(m)$ time complexity, where $m$ is the number of unreleased refobs in $Q$, it can suffer from poor locality: finding every $\Refob x A B$ such that $Q \vdash \Chain(x)$ requires tracing a path from $B$ to all of its potential inverse acquaintances using the chains of $\CreatedUsing$ facts.

One way to address this problem is to keep the $\CreatedUsing$ chains short, by having actors not keep the $\CreatedUsing$ fact in their knowledge set for long periods of time. In the extreme case, actors can immediately perform the \textsc{SendInfo} rule whenever they create a refob. This relieves the snapshot aggregator from dealing with $\CreatedUsing$ chains entirely, at the cost of increased control messages between actors.

Another interesting approach is for the snapshot aggregator to use a heuristic to find \emph{some} finalized subset, not necessarily the largest one. For our heuristic, notice that $Q \vdash \Chain(x)$ implies $Q \vdash \Unreleased(x)$ in \emph{any} set of snapshots $Q$. This motivates a new definition:
\begin{defi}
    $B$ \emph{potentially depends} on $A$ in $Q$ if $A = B$ or there is a sequence of one or more refobs $(\Refob{x_1}{A_1}{A_2}),\dots,(\Refob{x_n}{A_{n-1}}{A_n})$ where $A = A_1$ and $B = A_n$ and, for each $i < n$, $Q \vdash \Unreleased(\Refob{x_i}{A_i}{A_{i+1}})$.
\end{defi}
Notice that if $B$ depends on $A$, then $B$ also potentially depends on $A$; the latter is a coarser relation than the former.

Our heuristic algorithm is identical to the original, except that $S_2$ is the set of all actors that \emph{potentially} depend on $S_1$. Since the ``potentially depends'' relation is coarser than the ``depends'' relation, every snapshot in the resulting set is necessarily in the maximum finalized subset.

\begin{algorithm}
\caption{Compute a finalized subset of $Q$}\label{alg:fast-finalized-subset}
\begin{algorithmic}[1]
    \State Let $S_1 \subseteq \dom(Q)$ be the set of all actors $B$ for which there exists $\Refob x A B$ such that $Q \vdash \Chain(x)$ and $Q \not\vdash \Relevant(x)$.
    \State Let $S_2 \subseteq \dom(Q)$ be the set of all actors that \emph{potentially depend} on actors in $S_1$.
    \State Let $S_3 = \dom(Q) \setminus S_2$.
\end{algorithmic}
\end{algorithm}

The following lemma shows that, indeed, only ``stale'' snapshots prevent the resulting set from being the largest finalized subset.

\begin{lem}
    Let $Q$ be an arbitrary set of snapshots at time $t$, and $Q_f$ the largest finalized subset of $Q$. Let $Q'$ be another set of snapshots, all taken after time $t$, such that $\dom(Q')\cap\dom(Q) = \emptyset$.
    
    Then for all $B \in \dom(Q_f)$, for all $\Refob x A B$, $Q' \cup Q_f \vdash \Unreleased(\Refob x A B)$ implies $Q' \cup Q_f \vdash \Chain(\Refob x A B)$.
\end{lem}
\begin{proof}
    Since $Q_f$ is finalized, $Q_f \vdash \Unreleased(\Refob x A B)$ implies $Q_f \vdash \Chain(\Refob x A B)$. Moreover, $Q_f$ is a consistent closed snapshot at all times $t' \ge t$. Hence, for ant $B \in \dom(Q_f)$, if $\Refob x A B$ is unreleased at time $t'$ then $A \in \dom(Q_f)$.
    
    Now let $Q' \cup Q_f \vdash \Unreleased(\Refob x A B)$. By definition, this means $(Q' \cup Q_f)(B) \not\vdash \Released(x)$ and there exists some $C$ such that $(Q' \cup Q_f)(C) \vdash \Created(x)$. 
    
    If $C \in \dom(Q_f)$ then $Q_f \vdash \Unreleased(x)$ and therefore $Q_f \vdash \Chain(x)$ and therefore $Q' \cup Q_f \vdash \Chain(x)$. 
    
    Now suppose $C \in \dom(Q') \setminus \dom(Q_f)$. This implies $C \ne B$ and therefore $Q'(C) \vdash \CreatedUsing(y,x)$ for some $\Refob y C B$. This implies that $Q'(C) \vdash \Activated(y)$, so $y$ is unreleased at the time of $C$'s snapshot $t_C$. But since $\dom(Q_f)$ is closed at time $t_C$, this implies $C \in \dom(Q_f)$ after all; a contradiction. Hence $C \in \dom(Q_f)$, so $Q' \cup Q_f \vdash \Chain(x)$ by the argument above.
\end{proof}

Hence, if every non-terminated actor eventually takes a snapshot, a snapshot aggregator running the heuristic algorithm will eventually detect all terminated garbage.

\section{Cooperative Garbage Collection}\label{sec:coop}

Up to this point we have assumed the existence of a single snapshot aggregator that eventually receives all snapshots. However, there is no reason this must be a centralized entity. For instance, we can view a multicore actor system as a composition of $n$ actor systems; one for each processor core. It would be natural to have a snapshot aggregator for each system, dedicated to detecting and collecting terminated actors in that system. To detect cycles terminated sets of actors distributed across multiple systems, the aggregators can gossip their local snapshots amongst themselves; eventually every aggregator will obtain enough snapshots to detect all local terminated actors. Moreover, since the actor model is location-transparent, this same strategy extends to distributed multicore systems as well.

More formally, the cooperative garbage collection problem is for two snapshot aggregators, with disjoint snapshot sets $Q_1,Q_2$, to find maximal subsets $\hat Q_1 \subseteq Q_1$, $\hat Q_2 \subseteq Q_2$, such that $\hat Q_1 \cup \hat Q_2$ is finalized. For simplicity, we assume that neither $Q_1$ nor $Q_2$ has any finalized subsets, since such terminated actors could be detected without cooperation. Although we only consider the two-party case here, the discussion naturally generalizes to $n$ snapshot aggregators.

In this formalism, the simple strategy amounts to having the first aggregator send its entire snapshot set $Q_1$ to the second aggregator, and vice versa. This is clearly inefficient for two reasons. Firstly, the two aggregators must perform duplicate work to compute the maximum finalized subset of $Q_1 \cup Q_2$. Secondly, each snapshot set seems to contain significantly more information than is necessary to compute $\hat Q_1,\hat Q_2$; we might expect, for example, that it is only necessary to pass along snapshots from actors at the ``border'' of $Q_1,Q_2$ (e.g.~the receptionists).

In this section, we address both of the above concerns. We begin by defining \emph{potentially finalized} subsets of $Q_1$ and $Q_2$, which omit any snapshots that \emph{a priori} cannot be finalized in $Q_1 \cup Q_2$. Every actor in a potentially finalized set $Q$ depends on one or more of the receptionists in $Q$. Hence, computing $\hat Q_1,\hat Q_2$ reduces to finding the finalized receptionists of $Q_1,Q_2$. With this insight, we then show how to compute \emph{summaries} $\tilde Q_1,\tilde Q_2$ of $Q_1,Q_2$ such that the finalized receptionists in $\tilde Q_1 \cup \tilde Q_2$ coincide with those of $Q_1 \cup Q_2$. Aggregators can therefore simply exchange summaries to find the finalized receptionists. Since summaries can be significantly smaller than the original set of snapshots, this technique reduces the amount of data exchanged and reduces the amount of computation needed to detect finalized receptionists.

\subsection{Potentially finalized sets}

An actor $A$ in $Q_1$ could potentially be finalized in $Q_1 \cup Q_2$ if there exists $Q'$ disjoint from $Q_1$, such that $A$ is finalized in $Q_1 \cup Q'$. This motivates the following definition:
\begin{defi}
    $C$ is \emph{potentially finalized} in $Q$ if, for all $B$ on which $C$ depends, if $Q \vdash \Chain(\Refob x A B)$ then either $Q \vdash \Relevant(x)$ or $A \not\in \dom(Q)$.
\end{defi}
That is, $C$ would be finalized in $Q$ if it did not depend on some actors outside of $Q$. We say that a set $Q$ is potentially finalized if every actor in $Q$ is potentially finalized in $Q$.

Notice that if $C$ is \emph{not} potentially finalized in $Q$, then $C$ depends on some $B$ which has an irrelevant chain in $Q$. Such a $C$ is guaranteed not to be finalized in $Q_1 \cup Q_2$, for any $Q_2$. This means that any actor in $Q_i$ that is not potentially finalized in $Q_i$ can safely be removed from consideration, since it can neither be finalized in $Q_i$ nor $Q_1 \cup Q_2$.

Viewing $\dom(Q)$ as an actor system, we call $B$ a \emph{receptionist} in $Q$ if $B \in \dom(Q)$ and $Q \vdash \Chain(\Refob x A B)$ and $A \not\in \dom(Q)$. The following lemmas show that every $C \in Q_i$ depends on a receptionist.
\begin{lem}\label{lem:local-chain-relevant}
    If $A,B \in \dom(Q_i)$ and $Q_i \vdash \Chain(\Refob x A B)$ then $Q_i \vdash \Relevant(x)$.
\end{lem}
\begin{proof}
    Immediate from the assumption that $Q_i$ is potentially finalized.
\end{proof}
\begin{lem}\label{lem:all-depend-on-receptionists}
    Every $C \in \dom(Q_i)$ depends on some $B \in \dom(Q_i)$ where $A \not\in \dom(Q_i)$ and $Q_i \vdash \Chain(\Refob x A B)$.
\end{lem}
\begin{proof}
    Immediate from the assumption that $Q_i$ has no finalized subsets.
\end{proof}

Moreover, $C \in \dom(Q_i)$ is finalized in $Q_1 \cup Q_2$ if the receptionists on which it depends are finalized in $Q_1 \cup Q_2$:

\begin{lem}
    Let $A \in \dom(Q_1)$ and $B \in \dom(Q_2)$, without loss of generality. If $Q_1 \cup Q_2 \vdash \Chain(\Refob x A B)$, then $B$ is a receptionist in $Q_2$.
\end{lem}
\begin{proof}
    Let $(\Refob{x_1}{A_1}{B}),\dots,(\Refob{x_n}{A_n}{B})$ be the chain from $B$ to $x$ in $Q_1 \cup Q_2$. Since $A_n = A \in \dom(Q_1)$, there must be some $m \le n$ such that $\forall i < m,\ A_i \in \dom(Q_2)$ and $A_m \in \dom(Q_1)$. Hence $Q_2 \vdash \Chain(x_m)$ and $A_m \not\in \dom(Q_2)$, so $B$ is a receptionist in $Q_2$.
\end{proof}

\begin{lem}
    If $C \in \dom(Q_i)$ depends on $A$ in $Q_1 \cup Q_2$, then either \emph{(1)} $C$ depends on $A$ in $Q_i$, or \emph{(2)} $C$ depends on a receptionist $B$ in $Q_i$ and $B$ depends on $A$ in $Q_1 \cup Q_2$.
\end{lem}
\begin{proof}
    Let $(\Refob{x_1}{A_1}{A_2})$, $\dots$, $(\Refob{x_n}{A_{n-1}}{A_n})$ be the sequence of refobs from $A$ to $C$. If $\forall i \le n,\ A_i \in \dom(Q_i)$, then $C$ depends on $A$ in $Q_i$. Otherwise, let $m < n$ be the greatest index such that $A_m \not\in \dom(Q_i)$; then $A_{m+1} \in \dom(Q_i)$ is a receptionist that depends on $A$ and $C$ depends on $A_{m+1}$ in $Q_i$.
\end{proof}

\begin{thm}
    A non-receptionist $B \in \dom(Q_i)$ is finalized in $Q_1 \cup Q_2$ if and only if every receptionist on which $B$ depends in $Q_i$ is finalized in $Q_1 \cup Q_2$.
\end{thm}
\begin{proof}
    If $B \in \dom(Q_i)$ is finalized in $Q_1 \cup Q_2$ then every actor on which it depends must be finalized in $Q_1 \cup Q_2$. Since every actor on which $B$ depends in $Q_i$ is also depended upon in $Q_1 \cup Q_2$, the receptionists in particular must be finalized.
    
    Conversely, let $C \in \dom(Q_i)$ and let every receptionist on which $C$ depends in $Q_i$ be finalized in $Q_1 \cup Q_2$. We show that, if $C$ depends on $B$ in $Q_1 \cup Q_2$ and $Q_1 \cup Q_2 \vdash \Chain(\Refob x A B)$, then $Q_1 \cup Q_2 \vdash \Relevant(x)$. By the preceding lemma, there are two cases.
    
    \emph{Case 1.} $B\in \dom(Q_i)$ and $C$ depends on $B$ in $Q_i$. If $B$ is a receptionist of $Q_i$ then it is finalized by hypothesis; this implies $Q_1 \cup Q_2 \vdash \Relevant(x)$ \emph{a fortiori}. Otherwise, $A$ must be in $Q_i$, so $Q_i \vdash \Relevant(x)$ by Lemma~\ref{lem:local-chain-relevant}.
    
    \emph{Case 2.} $C$ depends on a receptionist $B'$ in $Q_i$ and $B'$ that depends on $B$ in $Q_1 \cup Q_2$. Then $B$ must be finalized because $B'$ is finalized.
\end{proof}
\begin{cor}
    A receptionist $B \in \dom(Q_2)$ is finalized in $Q_1 \cup Q_2$ if and only if $Q_1 \cup Q_2 \vdash \Chain(\Refob x A B)$ implies $Q_1 \cup Q_2 \vdash \Relevant(x)$ and $A$ is finalized.
\end{cor}

This formalizes our intuition that snapshots from ``internal actors'' of $Q_1$ and $Q_2$ are unnecessary. It suffices to combine the snapshots of actors at the ``boundary'' (e.g.~receptionists) with dependency information (i.e.~which ``boundary'' actors depend on which receptionists).

\subsection{Summaries}\label{sec:summary}

Based on the insight from the preceding section, our approach is to compute, for each $Q_i$, a smaller set of snapshots $\tilde Q_i$ called its \emph{summary}. These summaries are designed so that (1) all receptionists in $Q_i$ have snapshots in $\tilde Q_i$, and (2) a receptionist is finalized in $\tilde Q_1 \cup \tilde Q_2$ if and only if it is finalized in $Q_1 \cup Q_2$. We achieve this by removing all facts about the ``internal structure'' of each $Q_i$ and then adding new refobs to encode the dependency information of $Q_i$.

\begin{defi}
    The \emph{summary} $\tilde Q$ of $Q$ is the least set of snapshots satisfying the following properties:

For any $\Refob x A B$ where $A \in \dom(Q)$ and either $B$ is a receptionist or $B \not\in \dom(Q)$:
\begin{itemize}
    \item If $Q(A) \vdash \Activated(x)$ then $\tilde Q(A) \vdash \Activated(x)$;
    \item If $Q(A) \vdash \CreatedUsing(x,y)$ for some $y$ then $\tilde Q(A) \vdash \CreatedUsing(x,y)$.
    \item If $Q(A) \vdash \SentCount(x,n)$ then $\tilde Q(A) \vdash \SentCount(x,n)$;
\end{itemize}
For any $\Refob x A B$ where $B$ is a receptionist:
\begin{itemize}
    \item If $Q(B) \vdash \Created(x)$ then $\tilde Q(B) \vdash \Created(x)$;
    \item If $Q(B) \vdash \Released(x)$ then $\tilde Q(B) \vdash \Released(x)$;
    \item If $Q(B) \vdash \RecvCount(x,n)$ then $\tilde Q(B) \vdash \RecvCount(x,n)$;
\end{itemize}
If $A,B \in \dom(\tilde Q)$ and $A$ is a receptionist and $B$ depends on $A$, then  $\tilde Q(A) \vdash \Activated(x)$ and $\tilde Q(B) \vdash \Created(x)$ for some new, ``fake'' refob $\Refob x A B$ with a fresh token $x$.
\end{defi}

By this definition, both $Q_1 \cup Q_2$ and $\tilde Q_1 \cup \tilde Q_2$ agree about refobs $\Refob x A B$ where the owner is in $Q_1$ (resp.~$Q_2$) and the target is in $Q_2$ (resp.~$Q_1$):

\begin{lem}\label{lem:coop-chains-preserved}
    Let $A \in \tilde Q_1$ and $B \in \tilde Q_2$. For any $\Refob x A B$, if $Q_1 \cup Q_2 \vdash \Chain(x)$ then $\tilde Q_1 \cup \tilde Q_2 \vdash \Chain(x)$.
\end{lem}
\begin{proof}
    Let $Q_1 \cup Q_2 \vdash \Chain(x)$ and let $(\Refob{x_1}{A_1}{B}),\dots,(\Refob{x_n}{A_n}{B})$ be the chain from $B$ to $x$ in $Q_1 \cup Q_2$. Notice that $B$ is a receptionist in $Q_2$. Hence, by definition of the summary, $\tilde Q_2(B) \vdash \Created(x_1)$. We now show that $\tilde Q_1 \cup \tilde Q_2 \vdash \CreatedUsing(x_i,x_{i+1})$ for each $i < n$.
    
    If $A_i \in \dom(Q_2)$ then $Q_2(A_i) \vdash \Activated(x_i) \land \CreatedUsing(x_i,x_{i+1})$. Since $B$ is a receptionist in $Q_2$, $\tilde Q_2(A_i) \vdash \Activated(x_i) \land \CreatedUsing(x_i,x_{i+1})$.
    
    Otherwise, $A_i \in \dom(Q_1)$ and therefore $Q_1(A_i) \vdash \Activated(x_i) \land \CreatedUsing(x_i,x_{i+1})$. Since $B \not\in \dom(Q_1)$, $\tilde Q_1(A_i) \vdash \Activated(x_i) \land \CreatedUsing(x_i,x_{i+1})$.
\end{proof}

\begin{lem}\label{lem:coop-relevance-preserved}
    Let $A \in \tilde Q_1$ and $B \in \tilde Q_2$ such that $Q_1 \cup Q_2 \vdash \Chain(\Refob x A B)$. Then $Q_1 \cup Q_2 \vdash \Relevant(x)$ if and only if $\tilde Q_1 \cup \tilde Q_2 \vdash \Relevant(x)$.
\end{lem}
\begin{proof}
    Let $Q_1 \cup Q_2 \vdash \Relevant(x)$. Then there exists $n$ such that $Q_1(A) \vdash \Activated(x) \land \SentCount(x,n)$ and $Q_2(B) \vdash \RecvCount(x,n)$. Since $A \in \dom(Q_1)$ and $B \not\in \dom(Q_1)$, $\tilde Q_1(A) \vdash \Activated(x) \land \SentCount(x,n)$. Since $B$ is a receptionist in $Q_2$, $\tilde Q_2(B) \vdash \RecvCount(x,n)$.
    
    Conversely, let $\tilde Q_1 \cup \tilde Q_2 \vdash \Relevant(x)$. Then there exists $n$ such that $\tilde Q_1(A) \vdash \Activated(x) \land \SentCount(x,n)$ and $\tilde Q_2(B) \vdash \RecvCount(x,n)$. From the definition of $\tilde Q_1$ and $\tilde Q_2$, we must also have $Q_1(A) \vdash \Activated(x) \land \SentCount(x,n)$ and $Q_2(B) \vdash \RecvCount(x,n)$.
\end{proof}

The following lemma formalizes our understanding that the refobs in $\tilde Q$ serve to abbreviate the dependency information of $Q$:

\begin{lem}\label{lem:coop-summary-chains-cases}
    Let $A,B \in \tilde Q_1 \cup \tilde Q_2$. If $\tilde Q_1 \cup \tilde Q_2 \vdash \Chain(\Refob x A B)$ then either:
    \begin{enumerate}
        \item $Q_1 \cup Q_2 \vdash \Chain(\Refob x A B)$; or
        \item Both $A,B$ are in some $Q_i$ and $B$ depends on $A$ in $Q_i$.
    \end{enumerate}
\end{lem}
\begin{proof}
    Let $(\Refob{x_1}{A_1}{B}),\dots,(\Refob{x_n}{A_n}{B})$ be the chain from $B$ to $x$ in $\tilde Q_1 \cup \tilde Q_2$. Then for some $Q_i$, we must have $\tilde Q_i(B) \vdash \Created(x_1)$. 
    
    By construction of $\tilde Q_i$, this could either be the result of (1) $B$ being a receptionist in $Q_i$ or (2) $A_1$ being a receptionist in $Q_i$ and $B$ depending on $A_1$ in $Q_i$. 
    
    \emph{Case 1.} In this case, we must have $Q_i(B) \vdash \Created(x_1)$. Moreover, by construction of $\tilde Q_1$ and $\tilde Q_2$, $(\tilde Q_1 \cup \tilde Q_2)(A_j) \vdash \CreatedUsing(x_j,x_{j+1})$ implies $(Q_1 \cup Q_2)(A_j) \vdash \CreatedUsing(x_j,x_{j+1})$ for every $j < n$. Hence $Q_1 \cup Q_2 \vdash \Chain(x)$.
    
    \emph{Case 2.} In the latter case, the chain can only have length 1 because $x_1$ is a ``fake'' refob. Hence $x = x_1$ and $A_1 = A$, so indeed $B$ depends on $A$ in $Q_i$.
\end{proof}
\begin{cor}\label{lem:coop-summary-depends-implies}
    If $B$ depends on $A$ in $\tilde Q_1 \cup \tilde Q_2$ then $B$ depends on $A$ in $Q_1 \cup Q_2$.
\end{cor}

Conversely, we now show that all the important dependencies have been preserved - namely, which actors depend on which receptionists.

\begin{lem}\label{lem:coop-depends-implies-summary}
    Let $A,B \in \tilde Q_1 \cup \tilde Q_2$ and let $A$ be a receptionist in $Q_1$. If $B$ depends on $A$ in $Q_1 \cup Q_2$ then $B$ depends on $A$ in $\tilde Q_1 \cup \tilde Q_2$.
\end{lem}
\begin{proof}
    Let $(\Refob{x_1}{A_1}{A_2})$, $\dots$, $(\Refob{x_n}{A_{n-1}}{A_n})$ be the sequence of refobs from $A$ to $B$.
    
    If $n = 0$ then the lemma is trivially satisfied. 
    
    If $\forall i \le n,\ A_i \in \dom(Q_1)$ then, by construction of $\tilde Q_1$, there exists $\Refob x A B$ such that $\tilde Q_1 \vdash \Chain(x)$.
    
    For the general case, the sequence of refobs may pass between $Q_1$ and $Q_2$ multiple times. We partition the sequence $x_1,\dots,x_n$ into a sequence of ``runs'' $\vec x_1,\dots,\vec x_m$, such that:
    \begin{enumerate}
        \item For each refob $(\Refob x A B)$ in the first run $\vec x_1$, the owner $A$ is in $Q_1$; for each refob $(\Refob x A B)$ in the second run $\vec x_2$, the owner $A$ is in $Q_2$; for each refob $(\Refob x A B)$ in the third run $\vec x_3$, the owner $A$ is in $Q_1$; and so on.
        \item The concatenation of $\vec x_1,\dots,\vec x_m$ is $x_1,\dots,x_n$.
    \end{enumerate}
    For each run $\vec x_i$, we denote the owner of the first refob as $B_i$ and the target of the last refob as $C_i$. Notice that every $B_i$ is a receptionist. Hence, by construction of $\tilde Q_1$ and $\tilde Q_2$ there is, for each run $\vec x_i$, a refob $\Refob{y_i}{B_i}{C_i}$ such that $\tilde Q_1 \cup \tilde Q_2 \vdash \Chain(y_i)$. Then the sequence of refobs $y_1,\dots,y_m$ witnesses the fact that $B$ depends on $A$ in $\tilde Q_1 \cup \tilde Q_2$.
\end{proof}

Finally, we can show that summaries are sound and complete for the intended purpose of finding finalized receptionists.

\begin{thm}
    Let $C \in \tilde Q_1 \cup \tilde Q_2$. Then $C$ is finalized in $\tilde Q_1 \cup \tilde Q_2$ if and only if $C$ is finalized in $Q_1 \cup Q_2$.
\end{thm}
\begin{proof}
    Let $C$ be finalized in $\tilde Q_1 \cup \tilde Q_2$. We show that, if $C$ depends on some $B$ in $Q_1 \cup Q_2$ and $Q_1 \cup Q_2 \vdash \Chain(\Refob x A B)$, then $Q_1 \cup Q_2 \vdash \Relevant(x)$.
    
    If $A,B$ are both in some $Q_i$, then $Q_i \vdash \Relevant(x)$ because $Q_i$ is potentially finalized.
    
    Otherwise, let $A \in \dom(Q_1)$ and $B \in \dom(Q_2)$, without loss of generality. Since $Q_1 \cup Q_2 \vdash \Chain(x)$, Lemma~\ref{lem:coop-chains-preserved} implies $\tilde Q_1 \cup \tilde Q_2 \vdash \Chain(x)$. Since $B$ is a receptionist, $B \in \tilde Q_2$. Since $B,C \in \tilde Q_1 \cup \tilde Q_2$ and $C$ depends on $B$ in $Q_1 \cup Q_2$, $C$ must also depend on $B$ in $\tilde Q_1 \cup \tilde Q_2$ by Lemma~\ref{lem:coop-depends-implies-summary}. Hence, since $C$ is finalized, $\tilde Q_1 \cup \tilde Q_2 \vdash \Relevant(x)$. This implies, by Lemma~\ref{lem:coop-relevance-preserved}, $Q_1 \cup Q_2 \vdash \Relevant(x)$.
    
    ---
    
    Conversely, let $C$ be finalized in $Q_1 \cup Q_2$. We show that, if $C$ depends on some $B$ in $\tilde Q_1 \cup \tilde Q_2$ and $\tilde Q_1 \cup \tilde Q_2 \vdash \Chain(\Refob x A B)$, then $\tilde Q_1 \cup \tilde Q_2 \vdash \Relevant(x)$. By Lemma~\ref{lem:coop-summary-chains-cases}, there are two cases:
    
    \emph{Case 1:} $B$ depends on $A$ in $Q_i$; the refob $x$ is a ``fake'' reference created in the process of constructing $\tilde Q_i$. Then $\tilde Q_i(A) \vdash \Activated(x)$ and $\tilde Q_i \vdash \SentCount(x,0) \land \RecvCount(x,0)$ by construction.
    
    \emph{Case 2:} $Q_1 \cup Q_2 \vdash \Chain(x)$. Since $C$ is finalized and $C$ depends on $B$ in $Q_1 \cup Q_2$, we must have $Q_1 \cup Q_2 \vdash \Relevant(x)$. Then, since $B,C \in \tilde Q_1 \cup \tilde Q_2$, by Lemma~\ref{lem:coop-relevance-preserved}, it follows that $\tilde Q_1 \cup \tilde Q_2 \vdash \Relevant(x)$.
\end{proof}

Hence a pair of aggregators can find terminated actors in $Q_1 \cup Q_2$ by:
\begin{enumerate}
    \item Garbage collecting all finalized actors in each $Q_i$;
    \item Removing all actors not potentially finalized in each $Q_i$;
    \item Computing the summary of the remaining set of snapshots and exchanging it with their partner;
    \item Removing all potentially unfinalized snapshots in the pair of summaries $\tilde Q_1 \cup \tilde Q_2$ (optionally using the $\Unreleased$ heuristic from \Cref{sec:heuristic});
    \item Garbage collecting all actors in $Q_i$ that are reachable from a finalized receptionist in $\tilde Q_1 \cup \tilde Q_2$.
\end{enumerate}
Alternatively, a set of aggregators could send their summaries to a parent aggregator, which uses the summaries to compute the finalized receptionists and sends this set to each child aggregator.

\section{Conclusion and Future Work}\label{sec:conclusion}\label{sec:future-work}

We have shown how deferred reference listing and message counts can be used to detect termination in actor systems. The technique is provably safe (Corollary~\ref{cor:safety}) and live (\Cref{thm:liveness}). An implementation in Akka is presently underway.

We believe that DRL satisfies our three initial goals:
\begin{enumerate}
    \item \emph{Termination detection does not restrict concurrency in the application.} Actors do not need to coordinate their snapshots or pause execution during garbage collection.
    \item \emph{Termination detection does not impose high overhead.} The amortized memory overhead of our technique is linear in the number of unreleased refobs. Besides application messages, the only additional control messages required by the DRL communication protocol are $\InfoMsg$ and $\ReleaseMsg$ messages. These control messages can be batched together and deferred, at the cost of worse termination detection time.
    \item \emph{Termination detection scales with the number of nodes in the system.} Our algorithm is incremental, decentralized, and does not require synchronization between nodes.
\end{enumerate}
Since it does not matter what order snapshots are collected in, DRL can be used as a ``building block’’ for more sophisticated garbage collection algorithms. One promising direction is to take a \emph{generational} approach \cite{DBLP:journals/cacm/LiebermanH83}, in which long-lived actors take snapshots less frequently than short-lived actors. Different types of actors could also take snapshots at different rates. In another approach, snapshot aggregators could \emph{request} snapshots instead of waiting to receive them.
In the presence of faults, DRL remains safe but its liveness properties are affected. If an actor $A$ crashes and its state cannot be recovered, then none of its refobs can be released and the aggregator will never receive its snapshot. Consequently, all actors potentially reachable from $A$ can no longer be garbage collected. However, $A$'s failure does not affect the garbage collection of actors it cannot reach. In particular, network partitions between nodes will not delay node-local garbage collection.
Another issue that can affect liveness is message loss: If any messages along a refob $\Refob x A B$ are dropped, then $B$ can never be garbage collected because it will always appear unblocked. This is, in fact, the desired behavior if one cannot guarantee that the message will not be delivered at some later point. In practice, this problem might be addressed with watermarking.
Choosing an adequate fault-recovery protocol will likely vary depending on the target actor framework. One option is to use checkpointing or event-sourcing to persist GC state; the resulting overhead may be acceptable in applications that do not frequently spawn actors or create refobs. Another option is to monitor actors for failure and infer which refobs are no longer active; this is a subject for future work.
\section*{Acknowledgments}
\noindent This work was supported in part by the National Science Foundation under Grant No. SHF 1617401, and in part by the Laboratory Directed Research and Development program at Sandia National Laboratories, a multi-mission laboratory managed and operated by National Technology and Engineering Solutions of Sandia, LLC, a wholly owned subsidiary of Honeywell International, Inc., for the U.S. Department of Energy’s National Nuclear Security Administration under contract DE-NA0003525.
We would like to thank Dipayan Mukherjee, Atul Sandur, Charles Kuch, Jerry Wu, and the anonymous referees at CONCUR and LMCS for providing valuable feedback in earlier versions of this work.
\bibliography{contents/bibliography}{}
\bibliographystyle{alphaurl}
\end{document}